\documentclass[11pt, reqno, oneside]{amsart}

\usepackage{enumerate, microtype, amssymb, addlines}
\usepackage[pdftex]{color, graphicx}
\usepackage[longnamesfirst]{natbib}
\shortcites{dedeckeretal2007}
\defcitealias{ibragimovmueller2016}{Ibragimov-M\"uller}
\defcitealias{canayetal2014}{Canay-Romano-Shaikh}
\defcitealias{besteretal2014}{Bester-Conley-Hansen}

\input{macros}

\theoremstyle{plain}
\newtheorem{theorem}{Theorem}[section]
\newtheorem{lemma}[theorem]{Lemma}
\newtheorem{proposition}[theorem]{Proposition}
\newtheorem{assumption}[theorem]{Assumption}
\newtheorem{corollary}[theorem]{Corollary}
\theoremstyle{definition}
\newtheorem{algorithm}[theorem]{Algorithm}
\newtheorem{example}[theorem]{Example}
\theoremstyle{remark}
\newtheorem*{remarks}{Remarks}
\newtheorem*{remark}{Remark}

\newcommand{\rn}{\nu_n}
\newcommand{\ev}{{\mathord \mathrm{E}}}%
\newcommand{\prob}{{\mathord P}}%
\newcommand{\op}{o_\prob}%
\DeclareMathOperator*{\var}{Var}%
\DeclareMathOperator*{\trace}{trace}%
\DeclareMathOperator*{\rank}{rank}%
\newcommand{\pto}{\mathchoice
	{\raisebox{.0em}{ $\overset{\mathrm{P}}{\to}$ }}
	{\raisebox{-.15em}{ $\overset{\raisebox{-.25em}{\scriptsize$\mathrm{P}$}}{\to}$ }}
	{}
	{}}
\newcommand{\xqed}[1]{%
  \leavevmode\unskip\penalty9999 \hbox{}\nobreak\hfill
  \quad\hbox{\ensuremath{#1}}}
\newcommand{\sqed}{\xqed{\square}}
\newcommand{\mean}{\bar}
\newcommand{\smean}{\bar}

\usepackage{xcolor}
\definecolor{umorange}{HTML}{cc6600}
\definecolor{umblue}{HTML}{587abc}
\definecolor{umgrey}{HTML}{989c97}
\definecolor{umred}{HTML}{7a121c}
\definecolor{umgreen}{HTML}{83b2a8}

\usepackage[bookmarks=false,pdfpagemode=UseNone,pdftex,colorlinks,citecolor=black,%
            filecolor=black,linkcolor=black,urlcolor=black,%
            pdfauthor=Andreas~Hagemann,pdfstartview=FitH]{hyperref}

\begin{document}

\title[Placebo inference with clusters]{Placebo inference on treatment effects\\ when the number of clusters is small}
\author[]{Andreas Hagemann}
\address{Department of Economics, University of Michigan, 611 Tappan Ave, Ann Arbor, MI 48109, USA. Tel.: +1 (734) 764-2355. Fax: +1 (734) 764-2769}
\email{\href{mailto:hagem@umich.edu}{hagem@umich.edu}}
\urladdr{\href{https://umich.edu/~hagem}{umich.edu/~hagem}}
\date{\today. (First version online: November 14, 2017.)
}
\thanks{I would like to thank the co-editor, two referees, Federico Bugni, Matias Cattaneo, Kevin Lang, Sarah Miller, Azeem Shaikh, Elie Tamer, Matthew Webb, and several seminar audiences for useful comments and discussions. All errors are my own.}

\begin{abstract}
I introduce a general, Fisher-style randomization testing framework to conduct nearly exact inference about the lack of effect of a binary treatment in the presence of very few, large clusters when the treatment effect is identified across clusters. The proposed randomization test formalizes and extends the intuitive notion of generating null distributions by assigning placebo treatments to untreated clusters. I show that under simple and easily verifiable conditions, the placebo test leads to asymptotically valid inference in a very large class of empirically relevant models. Examples discussed explicitly are (i)~least squares regression with cluster-level treatment, (ii)~difference-in-differences estimation, and (iii)~binary choice models with cluster-level treatment. A simulation study and an empirical example are provided. The proposed inference procedure is easy to implement and performs well with as few as three treated and three untreated clusters.
\end{abstract}

\maketitle

\section{Introduction}
It is standard practice in economics to conduct inference that is robust to within-cluster dependence. Units in the same cluster have to be expected to influence one another or are influenced by the same technical, political, or environmental shocks. Several analytical and bootstrap procedures are available to adjust inference for the presence of data clusters. However, as \citet{bertrandetal2004} and others point out---typically in the context of inference on the treatment effect in a difference-in-differences model---the majority of these procedures perform poorly in situations where the number of clusters is small. Such situations are common in empirical practice. They arise, for example, in the analysis of policy reforms, where entire states are treated with the passage of a new law, or in a development context, where the introduction of a new technology affects entire villages. In this paper, I introduce a testing framework based on \citet{fisher1935} randomization that allows for nearly exact inference about the lack of effect of a treatment in the presence of a small number of large clusters. The number of clusters can either grow slightly with the sample size or remain fixed depending on the strength of other assumptions. The framework applies to situations where a binary treatment occurs in some but not all clusters and the treatment effect of interest is identified by between-cluster comparisons.

In a randomized trial, the average effect of a treatment is estimated by comparing the means of treatment and control groups. Computing this comparison of means for all possible ways in which individuals could have been assigned to the two groups generates ``placebo'' estimates. If treatment has no effect, the placebo estimates have the same distribution as the estimated treatment effect. A \citet{fisher1935} randomization or placebo test takes these observations as the null distribution to test the ``sharp'' hypothesis 
that there is no effect because the difference of treatment and control potential outcomes is zero for each individual. Such tests can be made exact under conventional assumptions. They are particularly attractive when only a small number of observations are available because the set of placebo estimates that have to be computed grows quickly with the sample size. More recently, placebo-type Monte Carlo experiments have been used in empirical economics as informal robustness exercises. I formalize and extend the notion of a placebo test to the cluster case by developing statistics that measure the size of a treatment effect of interest but are amenable to a placebo-like reassignment mechanism. Under simple and easily verifiable conditions, this placebo test leads to asymptotically valid, cluster-robust inference about conventional (non-sharp) null hypotheses in a very large class of empirically relevant models. The proofs rely in part on results of \citet{neuhaus1993} and \citet{janssen1997, janssen2005}, who show that consistent permutation tests of certain hypotheses are possible even in situations where the joint distribution of the data is not invariant to permutations under the null hypothesis. 

This paper complements recent work by \citet{canayetal2014} and \citet{ibragimovmueller2010, ibragimovmueller2016}. \citeauthor{canayetal2014}\ compute statistics for each cluster separately and obtain null distributions by permuting the signs of these statistics under an approximate symmetry assumption. The downside to their approach is that the parameter of interest has to be identified \emph{within} each cluster. Hence, clusters have to be paired in an ad-hoc manner for difference-in-differences estimation, which reduces the (already small) number of clusters available for inference by half. \citet{ibragimovmueller2010, ibragimovmueller2016} develop a method that applies to inference about parameters that are identified either within or across clusters; it involves comparing a summary of statistics from each cluster and invoking a small sample result for the $t$ distribution. The method of obtaining statistics used in \citet{ibragimovmueller2016} is similar (but not identical) to the one in the present paper. However, as \citeauthor{canayetal2014}\ point out, \citeauthor{ibragimovmueller2016}'s method tends to be overly conservative and can suffer from low power. As I show in my Monte Carlo study, the placebo test developed here tends to be less conservative and tends to have higher power than both the \citeauthor{canayetal2014}\ and \citeauthor{ibragimovmueller2016} tests when there are as few as three treated and three untreated clusters. In cases with six treated and six untreated clusters, all three methods and cluster-robust versions of the wild bootstrap \citep[see, e.g,][]{cameronetal2008, webb2013} give similar results. The placebo test is therefore especially appropriate for situations with very small numbers of clusters.
 
Other methods of cluster-robust inference are surveyed in \citet{cameronmiller2014} and \citet{mackinnonwebb2014}. They perform well for moderate numbers of clusters and are typically concerned with adjusting standard $t$ and $F$ tests in the linear regression model; \citet{donaldlang2007} derive corrections to standard errors and degrees of freedom under random-effects assumptions. \citet{hansen2007} and \citet{besteretal2014} use standard cluster-robust covariance matrix estimators in a framework with a small number of clusters but adjust critical values using techniques originally developed by \citet{kiefervogelsang2002,kiefervogelsang2005} for time series. \citet{imbenskolesar2016} derive degrees-of-freedom and standard error corrections following the approach of \citet{bellmccaffrey2002}. \citet{carteretal2013} develop measures that can be used to determine degrees-of-freedom corrections. Early papers that recognize the necessity of corrections include \citet{kloek1981} and \citet{moulton1990}. The present paper differs fundamentally from all of these approaches because it applies to a variety of models other than the linear regression model and derives critical values from the data that automatically account for within-cluster dependence instead of correcting the degrees of freedom of standard critical values.

The paper is organized as follows: Section~\ref{s:alg} constructs a general class of test statistics and shows how it can be used to conduct placebo inference. Section~\ref{s:asymptotics} establishes the asymptotic validity of the placebo test under explicit regularity conditions. Section~\ref{s:examples} verifies these conditions for several empirically relevant situations. Section~\ref{s:montecarlo} illustrates the finite sample behavior of the placebo test relative to other methods of inference in simulations, in data from the Current Population Survey, and in data from a large-scale experiment by \citet{dalbofrechette2011} on infinitely repeated games. Section~\ref{s:conc} concludes. The appendix contains auxiliary results and proofs.

I use the following notation: $1\{\cdot\}$ is the indicator function,  $|\cdot |$ is Euclidean norm, and $a\lesssim b$ means that $a$ is bounded by an absolute constant times $b$. For a matrix $A$, $|A|$ denotes matrix Euclidean norm $\sqrt{\trace(A'A)}$. The $L_p$-norm is $\Vert X\Vert_p = (\ev |X|^p)^{1/p}$. Unless otherwise noted, limits are as $n\to\infty$. Convergence in distribution is denoted by $\leadsto$.

\section{Placebo inference when the number of clusters is small}\label{s:alg}
This section introduces a general framework for Fisher-style placebo inference with a very small number of clusters. Applications of the framework to least squares regression with cluster-level treatment (Examples \ref{ex:clusterreg} and \ref{ex:clusterreg2} below) and difference-in-differences estimation (Examples \ref{ex:diffindiff} and \ref{ex:diffindiff2}) are provided.

Consider a situation where data from several large clusters (e.g., counties, regions, schools, firms, or stretches of time) are available. Observations are possibly dependent within clusters but are independent across clusters. Some of the clusters received treatment but others did not. The quantity of interest is a treatment effect or an object related to a treatment effect that can be represented by a scalar parameter~$\beta$. Because the entire cluster received treatment, this parameter is only determined up to a location shift $\theta_0$ within a treated cluster and only the left-hand side of \[ \theta_1 = \theta_0 + \beta \] can be identified from such a cluster. If the clusters have similar characteristics, $\theta_0$ can be identified from an untreated cluster. Comparing the two clusters identifies $\beta$. As the following two examples illustrate, this situation is prevalent in modern empirical work.

\addline

\begin{example}[Regression with cluster-level treatment]\label{ex:clusterreg} Consider a linear regression model
\begin{equation}\label{eq:regression}
Y_{i,k} = \theta_0 + \beta D_{k} + \eta_k' X_{i,k} + U_{i,k}, \qquad 1\leq i\leq m_{n,k}.
\end{equation}
Here $i$ indexes individuals within a cluster so that cluster $k$ has $m_{n,k}$ individuals. The goal is to conduct inference about the coefficient $\beta$ on the treatment dummy $D_k$ indicating whether cluster $k$ received treatment or not. In addition, the regression includes covariates $X_{i,k}$ that vary within each cluster and have coefficients $\eta_k$ that may vary across clusters such that $\ev(U_{i,k}\mid D_k, X_{i,k}) = 0$. Because $D_k$ is either $1$ or $0$, the data identify $\theta_1 = \theta_0 + \beta$ within a treated cluster and $\theta_0$ within an untreated cluster. \sqed
\end{example}

\begin{example}[Difference in differences]\label{ex:diffindiff}  Adjust the preceding example to the fixed-effects panel model
\begin{equation}\label{eq:diffindiff}
Y_{t,k} = \theta_0 I_t + \beta I_t D_{k} + \eta_k' X_{t,k} + \zeta_k + U_{t,k}.
\end{equation}
Now $k$ indexes individuals, $t$ indexes time, the dummy $I_t=1\{t>t_0\}$ indicates periods after an intervention at known time $t_0$, and the dummy $D_k$ indicates whether an intervention occurred. The $m_{n,k}$ observations 
of individual $k$ now form the $k$-th cluster and the $\zeta_k$ are cluster fixed effects. Provided the components of $X_{t,k}$ vary before or after $t_0$, the data again identify $\theta_1 = \theta_0 + \beta$ in a treated cluster and $\theta_0$ in an untreated cluster. \sqed
\end{example}

The goal of this paper is to develop a simple permutation test of the null hypothesis \[H_0\colon \beta=0 \] (or, equivalently, $H_0\colon \theta_1 = \theta_0$) that bases its decision on a null distribution obtained by assigning ``placebo treatments'' to all possible combinations of clusters and leads to asymptotically valid inference when only very few clusters are available. For the asymptotic theory, the number of treated clusters $q_{1,n}$ and the number of untreated clusters $q_{0,n}$ grow with the total sample size $n = \sum_{k=1}^{q_n} m_{n,k}$ (where $q_n = q_{1,n} + q_{0,n}$) but in some circumstances can be fixed if other conditions are strengthened. If $q_n$ grows with $n$, it does not have to do so at a specific rate. The number of clusters can therefore be very small relative to $n$. The parameter $\beta$ does not need to be interpretable by itself and may possibly only determine whether the actual treatment effect of interest is zero or not. For example, suppose \eqref{eq:regression} is the latent model in a binary choice framework with symmetric link function $F$ and $\eta_k \equiv \eta$;  then $F(\theta_0 + \beta + \eta' x) - F(\theta_0 + \eta' x)$ for some $x$ is typically the relevant treatment effect but $H_0\colon \beta=0$ is still the appropriate hypothesis to test. 

The permutation test developed in this paper is based on the idea that each cluster $k$ can provide an estimate $\hat{\theta}_{n,k}$ of either $\theta_1$ or $\theta_0$ depending on whether $k$ received treatment or not. Here, $\hat{\theta}_{n,k}$ can be any estimate that uses data only from cluster $k$ and satisfies a mild regularity condition (Assumption \ref{as:decom} in the next section). The condition essentially states that either $\sqrt{n/q_n}(\smash{\hat{\theta}_{n,k}} - \theta_1)$ or $\sqrt{n/q_n}(\smash{\hat{\theta}_{n,k}} - \theta_0)$ can be approximated by a well-behaved random variable with positive variance. These variances need not be identical across clusters, so the $\smash{\hat{\theta}_{n,k}}$ do not have to be standardized. Implicit in the assumption of a positive variance is also the requirement that the individual cluster sizes $m_{n,k}$ and the average cluster size $n/q_n$ grow at a similar rate (in the sense that $m_{n,k}q_n/n$ converges to a positive constant) to ensure that each cluster provides similarly good estimates of either $\theta_1$ or $\theta_0$. Convergence rates $\rn\to \infty$ other than $\sqrt{n/q_n}$ can also be accomodated as long as $\rn(\smash{\hat{\theta}_{n,k}} - \theta_1)$ or $\rn(\smash{\hat{\theta}_{n,k}} - \theta_0)$ have the required properties. Nonstandard estimators such as the smoothed maximum score estimator of \citet{horowitz1992}, for which $\rn$ could be as slow as $(n/q_n)^{2/5}$, are therefore also included in the analysis. 

I now give two examples of appropriate cluster-level statistics $\hat{\theta}_{n,k}$. In addition to the examples provided here, the methods discussed in this paper apply to estimates arising from, e.g., quantile regression, censored regression, and binary choice models such as probit. Section \ref{s:examples} contains further details and examples. 

\begin{example}[Regression with cluster-level treatment, continued]
\label{ex:clusterreg2} Suppose cluster $k$ received treatment and cluster $l$ did not. If each cluster is viewed as a separate regression, \eqref{eq:regression} can be written as
\[ Y_{i,k} = \theta_1 + \eta_k' X_{i,k} + U_{i,k}\qquad\text{and}\qquad Y_{i,l} = \theta_0 + \eta_l' X_{i,l} + U_{i,l}.\]
Denote the least squares estimates of the constants $\theta_1$ and $\theta_0$ in these regressions by $\hat{\theta}_{n,k}$ and $\hat{\theta}_{n,l}$. If $\beta = 0$, then $\theta_1 = \theta_0$ and $\smash{\hat{\theta}_{n,k} - \hat{\theta}_{n,l}} = \smash{(\hat{\theta}_{n,k} - \theta_0)} - \smash{(\hat{\theta}_{n,l} - \theta_0)}\approx 0$. Hence, even after rescaling by the convergence rates of the two least squares estimates, the difference
$\sqrt{m_{n,k}} (\hat{\theta}_{n,k} - \theta_0) - \sqrt{m_{n,l}} (\hat{\theta}_{n,l} - \theta_0)$
will be bounded in probability, whereas if $\theta_1 > \theta_0$ the display diverges to positive infinity as $n\to\infty$. Under standard assumptions, $\sqrt{m_{n,k}} (\hat{\theta}_{n,k} - \theta_0)$ and $\sqrt{m_{n,l}} (\hat{\theta}_{n,l} - \theta_0)$ have asymptotic linear representations with positive variance. The same must be true for $\sqrt{n/q_n} (\hat{\theta}_{n,k} - \theta_0)$ and $\sqrt{n/q_n} (\hat{\theta}_{n,l} - \theta_0)$ if $m_{n,k}q_n/n$ converges to a positive constant for every $1\leq k\leq q_n$.   \sqed
\end{example}

\begin{example}[Difference in differences, continued]\label{ex:diffindiff2} Suppose $k$ was treated but $l$ was not. View each cluster as a separate regression and rewrite \eqref{eq:diffindiff} as
\begin{equation*}
Y_{t,k} = \theta_1 I_t + \eta_k' X_{t,k} + \zeta_k + U_{t,k}\qquad\text{and}\qquad Y_{t,l} = \theta_0 I_t + \eta_l' X_{t,l} + \zeta_l + U_{t,l}.
\end{equation*}
The least squares estimates $\hat{\theta}_{n,k}$ and $\hat{\theta}_{n,l}$ of the slope parameters $\theta_1$ and $\theta_0$ are again suitable cluster-level estimates. \sqed
\end{example}

Order the data such that indices $1\leq k\leq q_{1,n}$ correspond to treated clusters and indices $q_{1,n}+1\leq k\leq q_n$ correspond to untreated clusters. Define the comparison-of-means function \[ (x_1, \dots, x_{q_n}) \mapsto \mean{T}_n(x_1, \dots, x_{q_n}) = \frac{1}{q_{1,n}}\sum_{k=1}^{q_{1,n}} x_k - \frac{1}{q_{0,n}}\sum_{k=1}^{q_{0,n}} x_{q_{1,n} + k} = \frac{1}{q_{1,n} q_{0,n}} \sum_{k=1}^{q_{1,n}}\sum_{l=1+q_{1,n}}^{q_n} \! \! \!\! (x_k-x_l). \] A comparison of means of the cluster estimates $\hat{\theta}_n = (\hat{\theta}_{n,1}, \dots, \hat{\theta}_{n,q_n})$, 
\begin{equation}\label{eq:tstat}
 T_n := \mean{T}_n(\hat{\theta}_n) = \frac{1}{q_{1,n} q_{0,n}} \sum_{k=1}^{q_{1,n}}\sum_{l=1+q_{1,n}}^{q_n} (\hat{\theta}_{n,k} - \hat{\theta}_{n,l}),
\end{equation}
summarizes all possible pairwise comparisons of treated and untreated clusters. The goal is now to construct a consistent permutation test based on $T_n$. This is possible because the scaled differences of cluster-level estimates $\rn(\hat{\theta}_{n,k} - \hat{\theta}_{n,l})$ remain stochastically bounded as $n$ grows under the null hypothesis $\beta=0$ and diverge to positive or negative infinity depending on whether $\beta > 0$ or $\beta < 0$. The results in Section \ref{s:asymptotics} show that the summary statistic $T_n$ inherits these features.

The idea underlying the permutation test based on \eqref{eq:tstat} is that if $H_0\colon \theta_1 = \theta_0$ is true, then the behavior of the estimates of $\theta_1$ and $\theta_0$ from any pair of clusters should be approximately the same---regardless of whether any of these clusters actually received treatment or not. Hence, \emph{all} differences $\hat{\theta}_{n,k} - \hat{\theta}_{n,l}$, $k\neq l\in\{1,2,\dots,q_n\}$, should behave similarly under the null. The statistic in \eqref{eq:tstat} could therefore be recomputed for every possible permutation of the indices $1,\dots,q_n$ and the distribution of these permuted statistics could be interpreted as an estimate of the null distribution of $T_n$. This is equivalent to assigning a placebo treatment indicator to every possible subset of $q_{1,n}$ clusters of the $q_n$ available clusters and recording the corresponding placebo realization of $T_n$. The original statistic $T_n$ could then be compared to the quantiles of the distribution of these placebo statistics. As I show in the next section, this intuition is correct in a very large class of models of practical interest if $q_{1,n}/q_n\to 1/2$ (see Corollary \ref{c:testconsistencyeq}) or in situations where the clusters are very similar (Theorem \ref{t:testconsistencyfin}). In general, however, the placebo realizations of $T_n$ have to be adjusted with a correction factor based on the two-sample variance function
\begin{align*}
(x_1, \dots, x_{q_n}) \mapsto \hat{S}_n^2(x_1, \dots, x_{q_n}) = &\frac{1}{q_{1,n}(q_{1,n} - 1)}\sum_{k=1}^{q_{1,n}} \biggl(x_k - \frac{1}{q_{1,n}}\sum_{l=1}^{q_{1,n}} x_l\biggr)^2\\ &\qquad + \frac{1}{q_{0,n}(q_{0,n} - 1)}\sum_{k=1+q_{1,n}}^{q_n} \biggl(x_k - \frac{1}{q_{0,n}}\sum_{l=1+q_{1,n}}^{q_n} x_k\biggr)^2
\end{align*} 
to yield a consistent test (Theorem \ref{t:testconsistency}). The reason for this adjustment is that a na\"ive permutation test would require the joint distribution of $\hat{\theta}_n = (\hat{\theta}_{n,1}, \dots, \hat{\theta}_{n,q_n})$ to be---at least asymptotically---invariant to permutation, which is not necessarily satisfied here.

To define the placebo statistics, let $\Pi^*$ be the set of permutations of $(1, 2, \dots, q_n)$. Each permutation $\pi\in \Pi^*$ depends on $n$ whenever $q_n$ does, but this is suppressed in the notation to prevent clutter. View the cluster-level estimates as maps $k\mapsto \hat{\theta}_{n,k}$ and, for each $\pi$, take $\pi(k)$ to be the $k$-th coordinate of $\pi$ so that $\hat{\theta}_{n,\pi(k)}$ is computed with data from cluster $\pi(k)$. Denote the permuted vector of cluster-level estimates by \[\pi\mapsto \pi \hat{\theta}_n = (\hat{\theta}_{n,\pi(1)}, \dots, \hat{\theta}_{n,\pi(q_n)}),\qquad \pi\in\Pi^*. \]
This gives rise to the family of (adjusted) placebo statistics
\begin{equation}\label{eq:permstat}
 T_n(\pi \hat{\theta}_n) = \mean{T}_n(\pi \hat{\theta}_n)\frac{\hat{S}_n(\hat{\theta}_n)}{\hat{S}_n(\pi \hat{\theta}_n)}
\end{equation}
indexed by $\pi\in\Pi^*$. Also note that $T_n(\hat{\theta}_n) = \mean{T}_n( \hat{\theta}_n) = T_n$. I will occasionally use $T_n(\hat{\theta}_n)$ and $\mean{T}_n( \hat{\theta}_n)$ to emphasize the dependence of $T_n$ on $\hat{\theta}_n$.

By construction, the ordering of the numbers $\pi(1),\dots, \pi(q_{1,n})$ and $\pi(q_{1,n} + 1),\dots, \pi(q_n)$ does not change the value of $T_n(\pi\hat{\theta}_n)$. Hence, it suffices to compute the $\binom{q_n}{q_{1,n}}$ placebo statistics indexed by the set $\Pi\subset \Pi^*$ for which the combination $\{\pi(1),\dots, \pi(q_{1,n})\}$ (and therefore the combination $\{\pi(q_{1,n}+ 1),\dots, \pi(q_n)\}$) 
is unique. One way of representing this set is \[ \Pi = \Bigl\{\pi\in\Pi^* : \pi(1) < \dots < \pi(q_{1,n})\text{~and~}\pi(q_{1,n}+1) < \dots < \pi(q_n)\Bigr\}. \]
Denote by $\lceil a \rceil$ the smallest integer larger than $a$ and let $|A|$ denote cardinality of a set $A$. The $1-\alpha$ quantile \[ c_{n,\alpha} = c_{n,\alpha}(\hat{\theta}_n) \] of $T_n(\pi\hat{\theta}_n)$ as $\pi$ varies over $\Pi$ is then the $\lceil |\Pi|(1-\alpha)\rceil$-th largest element of $\{T_n(\pi\hat{\theta}_n) : \pi\in\Pi\}$.

The following procedure provides a generic test for placebo inference with cluster-level statistics. Throughout this paper I consider one-sided tests against the alternative $\beta > 0$ but all results below remain valid for the alternatives $\beta < 0$ and $\beta \neq 0$. See the Remarks immediately below for the necessary modifications.

\begin{algorithm}[Placebo test]\phantomsection 
\label{al:placebo}
\begin{enumerate}[(i)]
	\item Compute $T_n$ as in \eqref{eq:tstat}.
	\item For each permutation $\pi\in\Pi$, compute $T_n(\pi\hat{\theta}_n)$ as in \eqref{eq:permstat}. 
	\item Reject the null hypothesis $\beta = 0$ in favor of the alternative $\beta > 0$ if $T_n$ exceeds $c_{n,\alpha}$, the $1-\alpha$ quantile of $\{T_n(\pi\hat{\theta}_n) : \pi\in\Pi\}$.
	\end{enumerate}
\end{algorithm}

\begin{remarks}
(i)~The adjustment factor in \eqref{eq:permstat} converges to one if $q_{1,n}/q_n\to 1/2$. For the placebo test it is therefore sufficient to compare $T_n$ to $\mean{c}_{n,\alpha} = \mean{c}_{n,\alpha}(\hat{\theta}_n)$, the $1-\alpha$ quantile of $\{\mean{T}_n(\pi\hat{\theta}_n) : \pi\in\Pi\}$, if $q_{1,n} = q_{0,n}$. The remaining parts of this Remark also apply to this version of the placebo test with $\mean{c}_{n,\alpha}$ in place of $c_{n,\alpha}$.

(ii)~The test decision can also be made with the \emph{p}-value 
\begin{equation}\label{eq:pval}
p_n(\hat{\theta}_n) =  \inf\{ p \in (0,1) : T_n(\hat{\theta}_n) > c_{n,p}(\hat{\theta}_n) \} = \frac{1}{|\Pi|}\sum_{\pi\in\Pi}1\{ T_{n}(\pi\hat{\theta}_n) \geq T_n(\hat{\theta}_n) \}
\end{equation}
because the statement $T_n > c_{n,\alpha}$ is equivalent to $p_n(\hat{\theta}_n) \leq \alpha$. See Appendix \ref{s:proofs} for proofs of this assertion and the second equality in the preceding display.

(iii)~For a one-sided test against $\beta < 0$, reject if $T_n(-\hat{\theta}_n) >  c_{n,\alpha}(-\hat{\theta}_n)$ or, equivalently, if  $T_n(\hat{\theta}_n)$ is smaller than the $\lfloor |\Pi|\alpha\rfloor$-th largest element of $\{T_n(\pi\hat{\theta}_n) : \pi\in\Pi\}$, where $\lfloor a \rfloor$ is the largest integer smaller than $a$. For a two-sided test against $\beta \neq 0$, reject if $T_n(\hat{\theta}_n) > c_{n,\alpha/2}(\hat{\theta}_n)$ or $T_n(-\hat{\theta}_n) > c_{n,\alpha/2}(-\hat{\theta}_n)$. A $p$-value for a two-sided test can be defined as $2\min\{ p_n(\hat{\theta}_n), p_n(-\hat{\theta}_n)\}.$

(iv)~From a theoretical standpoint, there is no difference in tests based on $\Pi$ or $\Pi^*$ because they lead to identical test decisions at any sample size. However, in practice the permutation statistics should always be computed from $\Pi$ in order to reduce the number of computations to $\binom{q_n}{q_{1,n}}$ statistics. For example, if there are $6$ treated and $6$ untreated clusters, then computing all permutations from $\Pi^*$ requires evaluation of $12! \approx 497$ million statistics whereas $\Pi$ requires only $\binom{12}{6} = 924$. If $\Pi$ is too large to make computations feasible, then $\Pi$ can be replaced with a set $\Pi_M$ consisting of $M$ draws from the uniform distribution on $\Pi$. It is easy to see that the resulting $p$-value $p_{n,M}(\hat{\theta}_n)$ approximates $p_n(\hat{\theta}_n)$ with arbitrary precision for $M$ large enough in the sense that $p_{n,M}(\hat{\theta}_n) \pto p_n(\hat{\theta}_n)$ as $M\to\infty$. \phantomsection\label{rev:randomperms}

(v)~If desired, the test decision in Algorithm \ref{al:placebo} can be based on a randomized test instead of the nonrandomized test $t\mapsto 1\{t > c_{n,\alpha}\}$. The randomized test can be constructed with the help of a test function $t\mapsto 1\{\varphi_{n,\alpha}(t)\geq U\}$ defined by an independent variable $U$ with a uniform distribution on $[0,1]$ and the nonrandomized test function 
\begin{equation} \label{eq:randtest}
\varphi_{n,\alpha}(t) = \begin{cases} 1 & t > c_{n,\alpha}\\ \delta_n & t = c_{n,\alpha}\\ 0 & t < c_{n,\alpha}, \end{cases}\quad \text{where }\delta_n = \frac{|\Pi|\alpha - |\{\pi\in\Pi : T_n(\pi\hat{\theta}_n) > c_{n,\alpha}\}|}{|\{\pi\in\Pi : T_n(\pi\hat{\theta}_n) = c_{n,\alpha}\}|}. 
\end{equation}
The randomized test has the advantage that if the distribution of the $\hat{\theta}_{n,1}, \dots, \hat{\theta}_{n,q_n}$ were invariant to permutation under the null hypothesis, then this test would be exact because $\prob (\varphi_{n,\alpha}(T_n) \geq U) = \alpha$ by a standard argument due to \citet{hoeffding1952}. However, such a test is of little use in practice because rejecting the null if $\varphi_{n,\alpha}(T_n)\geq U$ bases the test decision on a \emph{single} draw from the uniform distribution. Two researchers with identical data sets could therefore arrive at opposite test decisions simply because of two different draws of $U$ or could draw until a desired conclusion was reached.

(vi)~It should be noted that, under extreme circumstances, the power of a test that rejects if $T_n$ exceeds a critical value obtained from a permutation distribution can be zero. This is because $T_n\in \{T_n(\pi\hat{\theta}_n) : \pi\in\Pi\}$ and $T_n$ cannot be larger than $c_{n,\alpha}$ whenever $\lceil |\Pi|(1-\alpha)\rceil = |\Pi|$ or, equivalently, $|\Pi| < \alpha^{-1}$. (This does not change if $\Pi^*$ is used instead of $\Pi$.) Hence, if $\alpha=.05$, one needs $3$ treated and $3$ untreated clusters to ensure that $|\Pi| = \binom{6}{3} = 20 \geq \alpha^{-1}$ in order for Algorithm \ref{al:placebo} to have nonzero power. In contrast, the test of \citet{canayetal2014} has $2^{\min\{q_1,q_0\}}$ possible permutations and therefore their nonrandomized test needs at least $5$ treated and $5$ untreated clusters in order to have nonzero power with a significance level of $.05$. This problem persists in randomized versions of these tests because the test function $\varphi_{n,\alpha}$ in the preceding display is then bounded above by $\delta_n$, which in turn is strictly smaller than 1 if $|\Pi| < \alpha^{-1}$.\sqed
\end{remarks}

The next two sections prove the consistency of Algorithm \ref{al:placebo} under explicit regularity conditions and explain how these conditions can be verified. Section~\ref{s:montecarlo} illustrates the finite sample behavior of the placebo test relative to other methods of inference in three Monte Carlo experiments and an empirical application.

\section{Assumptions and generic asymptotics}\label{s:asymptotics}
This section introduces a simple high-level condition (Assumption \ref{as:decom}) that I use in Theorem \ref{t:testconsistency} to prove the validity of the placebo test. I then discuss a version of the placebo test that---under stronger assumptions---can be consistent with a fixed number of clusters.

The standard proof strategy for establishing weak convergence of an estimator is to decompose it into a well-behaved leading term (to which distributional limit theory is applied) and a small remainder. I assume that such a decomposition is available for each of the cluster-level statistics $\hat{\theta}_{n,k}$. I discuss this condition in detail immediately below. The next section shows how it can be verified for several standard estimators.
\begin{assumption}\label{as:decom}
There are constants  $\theta_1, \theta_0$ and a sequence $\rn\to\infty$ such that 
\begin{equation}\label{eq:decom}
\rn(\hat{\theta}_{n,k} - \theta_{1\{k\leq q_{1,n}\}}) = Z_{n,k} + R_{n,k}, \qquad 1\leq k\leq q_n,
\end{equation} 
where {\upshape (i)}~$\min\{q_{1,n},q_{0,n}\} \to\infty$ and $q_{1,n}/q_n\to \lambda\in (0,1)$, {\upshape (ii)}~$\ev Z_{n,k} = 0$ for each $k$, {\upshape (iii)}~$\ev |Z_{n,k}|^p$ is uniformly bounded in $n$ and $k$ for some $p > 2$, {\upshape (iv)}~$Z_{n,1},\dots, \smash{Z_{n,q_n}}$ are independent, {\upshape (v)}~$\smash{\sum_{k=1}^{q_n}} \var Z_{n,k}/q_n$ is bounded away from zero, and {\upshape (vi)}
\begin{equation}\label{eq:unegli}
\sum_{k=1}^{q_n} |R_{n, k}| = \op(\sqrt{q_n}).
\end{equation}
\end{assumption}

The leading term $Z_{n,k}$ in Assumption \ref{as:decom} is typically a byproduct of studying the large sample distribution of the estimator on which the $\hat{\theta}_{n,k}$ are based. Conditions (ii)-(v) impose a simple set of restrictions on this term. In particular, if $\smash{\hat{\theta}_{n,k}}$ is an estimator with an asymptotic linear representation (such as any of the estimators mentioned in the previous section), then $Z_{n,k}$ plays the role of that representation and $\ev |Z_{n,k}|^p$ in (iii) can be simply bounded further by inequalities for dependent data; see, among many others, \citet{wu2005b} and \citet{machkouriaetal2013} for easy-to-use results that apply to very general dependence structures. If the cluster structure is generated by time or spatial dependence, an informative bound on $\ev |Z_{n,k}|^p$ will be related to the summability of the auto-covariance function of the asymptotic linear representation. The condition on $\ev |Z_{n,k}|^p$ then guarantees that the data within all clusters exhibit short-range dependence, even though the precise form of the dependence does not need to be known to the researcher. This condition can be expected to fail, e.g., if individual time series in a panel (where cluster $k$ contains the time series of the $k$-th individual) contain unit roots. Condition (iv) imposes that data across clusters are independent. Condition (v) ensures that the leading terms in \eqref{eq:decom} generally do not have a degenerate distribution. It is \emph{not} required that the $Z_{n,k}$ converge in distribution, have the same variance, or that estimates of their variances are available. 

Conditions (i) and (vi) are the only restrictions on the total number of clusters $q_n$ and the remainder terms $R_{n,k}$. Together, they impose that the absolute sum of the remainder terms $R_{n,k}$ is small as $q_n$ grows. Neither a specific rate of divergence of $q_n$ nor closed-form expressions for $R_{n,k}$ are needed. However, as it stands, condition (vi) requires knowledge about the behavior of the collection of remainders $R_{n,1},\dots,R_{n,q_n}$ as $q_n$ grows, whereas in applications it is typically only known that $R_{n,k}\pto 0$ for each fixed $k$. Fortunately, as the following result shows, convergence at each $k$ is already enough to guarantee the existence of sequences $q_{1,n}$ and $q_{0,n}$ that satisfy conditions (i) and (vi). Inspection of its proof reveals that such sequences necessarily grow slowly with the total sample size.  This can be gleaned from the fact that if $q_n$ were fixed, convergence of the remainder terms to zero would imply \eqref{eq:unegli} by the continuous mapping theorem. Hence, condition (vi) is particularly appropriate in the present situation where the number of clusters is small relative to the sample size. 

\begin{proposition}\label{p:uniformremainder}
Suppose that $R_{n,k}\pto 0$ for each $k$. Then there are sequences $q_{1,n}\to \infty$ and $q_{0,n}\to \infty$ such that $q_{1,n}/(q_{1,n}+q_{0,n})\to \lambda\in (0,1)$ as $n\to\infty$ and \eqref{eq:unegli} holds.
\end{proposition}
\begin{remark}
\phantomsection\label{rev:remainderremark}%
The model at hand determines how large $q_n$ can be relative to other quantities. For example, I show in the next section that \eqref{eq:unegli} holds in a linear model with very general weak dependence assumptions under the mild condition $q_n/\inf_k m_{n,k}^2\to 0$. (Recall that $m_{n,k}$ is the size of the $k$-th cluster.) \sqed
\end{remark}

The following result establishes the consistency of the placebo test introduced in Algorithm~\ref{al:placebo} under Assumption \ref{as:decom}. The theorem is stated as a one-sided test to the right but remains valid with obvious modifications as a one-sided test to the left or as a two-sided test.
\begin{theorem}\label{t:testconsistency}
Suppose Assumption \ref{as:decom} is satisfied. For all $\alpha\in(0,1)$,
\begin{enumerate}[\upshape (i)]
\item if $\beta = 0$, then $\prob(T_n > c_{n,\alpha}) \to \alpha$ and
\item if $\beta > 0$, then $\prob(T_n > c_{n,\alpha})\to 1$.
\end{enumerate}
\end{theorem}

\begin{remark}
The event $\{T_n > c_{n,\alpha}\}$ is invariant to multiplication of both $T_n$ and the placebo statistics $\{T_n(\pi\hat{\theta}_n) : \pi\in\Pi\}$ by the same positive scalar. The proof of Theorem \ref{t:testconsistency} exploits this invariance to apply limit theory to properly scaled versions of both $T_n$ and $T_n(\pi\hat{\theta}_n)$. The main issue is the dependence structure induced by the permutation procedure. The proof of Theorem \ref{t:testconsistency}(i) also appeals to a remarkable result due to \citet{janssen2005}, who shows that consistent permutation tests of certain hypotheses are possible even in situations where the joint distribution of the data is not invariant to permutations under the null hypothesis. \citeauthor{janssen2005}'s results in turn rely on an insight of \citet{neuhaus1993}. \sqed
\end{remark}

If the number of treated and untreated clusters is approximately the same in the sense that $q_{1,n}/q_n\to 1/2$, then $\hat{S}_n(\hat{\theta}_n)/\hat{S}_n(\pi\hat{\theta}_n)\pto 1$. In that case it suffices to compare $T_n$ to $\mean{c}_{n,\alpha}$, the $1-\alpha$ quantile of $\{\mean{T}_{n}(\pi \hat{\theta}_n) : \pi\in\Pi\}$. (See the discussion above equation \eqref{eq:tstat} for definitions.)
A similar observation was made by \citet{romano1989} in a related context. 
This theoretical result is also confirmed by the Monte Carlo simulations in Section \ref{s:montecarlo}, which suggest that estimating the adjustment factor can be safely avoided if $q_{1,n} = q_{0,n}$.
\begin{corollary}\label{c:testconsistencyeq}
In the situation of Theorem \ref{t:testconsistency} with $q_{1,n}/q_n\to 1/2$,
\begin{enumerate}[\upshape (i)]
\item if $\beta = 0$, then $\prob(T_n > \mean{c}_{n,\alpha}) \to \alpha$ and
\item if $\beta > 0$, then $\prob(T_n > \mean{c}_{n,\alpha})\to 1$.
\end{enumerate}
\end{corollary}

I now briefly discuss the situation where the number of treated clusters $q_{1,n} = q_1$ and untreated clusters $q_{0,n} = q_0$ is fixed. As the following results shows, the result in Theorem~\ref{t:testconsistency} under the alternative remains unchanged as long as the significance level is not too small. (See the Remarks below Algorithm~\ref{al:placebo} for a discussion why the power is zero if $\alpha < |\Pi|^{-1}$.)

\begin{corollary}\label{c:testpowerfin}
Suppose there are constants  $\theta_1$ and $\theta_0$ such that $\rn(\hat{\theta}_{n,k} - \theta_{1\{k\leq q_1\}}) = Z_{n,k} + R_{n,k}, 1\leq k\leq q,$ where {\upshape (i)}~$\ev Z_{n,k} = 0$ for each $k$, {\upshape (ii)}~$\ev |Z_{n,k}|^p$ is uniformly bounded in $n$ and $k$ for some $p > 2$, {\upshape (iii)}~$Z_{n,1},\dots, \smash{Z_{n,q}}$ are independent, {\upshape (iv)}~$\var Z_{n,k}$ is bounded away from zero, $1\leq k\leq q$, and {\upshape (vi)} $\sum_{k=1}^{q} |R_{n, k}| \pto 0.$ If $\beta > 0$, then $\prob(T_n > c_{n,\alpha})\to 1\{ \alpha \geq |\Pi|^{-1} \}$.
\end{corollary}

\noindent However, with fixed $q_{1,n} = q_1$ and $q_{0,n} = q_0$, Assumption \ref{as:decom} does not seem to put enough conditions on the data to analyze the placebo test under the null hypothesis. To construct proper conditions, assume instead the joint convergence
\begin{equation}\label{eq:jointconv}
\rn(\hat{\theta}_{n,1} - \theta_{1},\dots, \hat{\theta}_{n,q_1} - \theta_{1},\hat{\theta}_{n,q_1+1} - \theta_{0},\dots, \hat{\theta}_{n,q} - \theta_{0}) \leadsto (Z_1,\dots,Z_{q}),
\end{equation}
\phantomsection\label{rev:behrensfishdisc}%
where the $Z_1,\dots,Z_{q}$ are independent. In the canonical situation where $Z = (Z_1,\dots,Z_{q})$ is mean-zero independent normal with standard deviations $\sigma = (\sigma_1, \dots, \sigma_q)$, extensive numerical computations suggest that $\sup_{\sigma\geq 0}\prob_\sigma(T_n(Z) > c_{n,\alpha}(Z))$ slightly exceeds $\alpha$. For example, for $q_1 = q_0 = 5$, I was unable to produce a rejection frequency above $.077$ for a test with $\alpha = .05$. This worst-case rejection frequency was achieved by combinations of near-infinite and zero variances. Less extreme values of $\sigma$ produced tests that were, for all practical purposes, correct. Heavily imbalanced $q_1$ and $q_0$ lead to worse performance but, as can be expected from Theorem \ref{t:testconsistency}, the worst-case rejection frequencies were lower with larger $q_1$ and $q_0$. For $q_1 = q_0 = 7$, they did not exceed $.065$ for tests with a $.05$ nominal level.

If the clusters are comparable enough such that the $(Z_1,\dots,Z_{q})$ in \eqref{eq:jointconv} are iid with a smooth but otherwise arbitrary distribution, then validity of the placebo test for fixed $q$ can be restored in the sense that the placebo test is asymptotically conservative. Under the alternative, the same conditions are enough to guarantee that the test has power. The results here are given for the unadjusted critical values $\mean{c}_{n,\alpha}$. This is because the underlying asymptotics no longer rely on a normal approximation, which makes variances insufficient to describe the behavior of $(\hat{\theta}_{n,1}, \dots, \hat{\theta}_{n,q})$. Standardization is therefore no longer useful. 

\begin{theorem}\label{t:testconsistencyfin}
Suppose that \eqref{eq:jointconv} holds and the $Z_k$ are iid with continuous distribution. For all $\alpha\in(0,1)$,
\begin{enumerate}[\upshape (i)]
\item if $\beta = 0$ , then $\limsup_{n\to\infty}\prob(T_n > \mean{c}_{n,\alpha}) \leq \alpha$ and
\item if $\beta > 0$, then $\prob(T_n > \mean{c}_{n,\alpha})\to 1\{ \alpha \geq |\Pi|^{-1} \}$.
\end{enumerate}
\end{theorem}

\begin{remark}
(i)~The proof of part (i) of the theorem relies heavily on results of \citet{canayetal2014}, who deal with parameters that are identified within clusters. Their results focus on symmetric data. To the best of my knowledge, the results in part (ii) of the theorem and Corollary \ref{c:testpowerfin} are new.

(ii)~See, e.g., Lemma 1 of \citet{besteretal2014} for a result that is appropriate for establishing the distributional convergence \eqref{eq:jointconv}. Results of this type usually require additional assumptions on how the space evolves as the number of data points increases and on the location of the data on that space. As the examples in the next section show, this difficulty can be avoided when working with Assumption \ref{as:decom}.

(iii)~If $q_1 = q_0$, then Corollary \ref{c:testconsistencyeq} and Theorem \ref{t:testconsistencyfin} together imply that the placebo test is ``doubly robust'' in the sense that it is consistent if the conditions of either Corollary \ref{c:testconsistencyeq} or Theorem \ref{t:testconsistencyfin} are met. 

\phantomsection\label{rev:optimalitydisc}%
(iv)~If $\alpha$ is a multiple of $1/q!$, the test $T_n(Z) > \mean{c}_{n,\alpha}(Z)$ with $Z$ as in Theorem \ref{t:testconsistencyfin} is uniformly most powerful similar against $Z_1,\dots, Z_{q_1} \sim N(\mu_1, \sigma^2)$ and $Z_{q_1+1},\dots, Z_{q} \sim N(\mu_0, \sigma^2)$ with $\mu_1 > \mu_0$, and $\sigma^2$ unknown. This ``optimality'' property follows from \citet{lehmannstein1949}. If $\alpha$ is not a multiple of $1/q!$, then randomization can restore similarity. The tests of \citet{ibragimovmueller2010} and \citet{canayetal2014} also have optimality properties in a one-sample setting for similarly restrictive classes of alternatives. \sqed
\end{remark}

\section{Examples}\label{s:examples}

This section shows how the central conditions of Assumption \ref{as:decom} can be established in a given application. Examples treated explicitly are least squares regression (Example \ref{ex:clusterreg3}) and binary choice (Example \ref{ex:binreg}) with cluster-level treatment, although the techniques described in this section apply more generally. 

For simplicity, I represent each cluster $k$ by a stationary random field \[\xi_{s,k} = \xi_k(\varepsilon_{s-h,k} \colon h\in\mathbb{Z}^d), \qquad s\in\mathbb{Z}^d, \] where fields $(\xi_{s,k})_{s\in\mathbb{Z}^d}$ and $(\xi_{s,l})_{s\in\mathbb{Z}^d}$ are independent whenever $k\neq l$. Here, $(\varepsilon_{s,k})_{s\in\mathbb{Z}^d}$ is a field of iid copies of a random variable $\varepsilon$ and $\xi_k$ is a measurable, possibly unknown function with values in $\mathbb{R}^{d_k}$ that transforms the input $(\varepsilon_{s,k})_{s\in\mathbb{Z}^d}$ into the output $\xi_{s,k}$. A very large class of commonly used time series ($d=1$) and spatial models are of this form; see \citet{shaowu2007} and \citet{machkouriaetal2013}. Many other dependence structures are possible \citep[see][and the references therein]{dedeckeretal2007}, but this type of random field leads to particularly simple conditions.

For the class of fields given in the preceding display, the dependence within each cluster can be measured by comparing $\xi_{s,k}$ to a slightly perturbed version of itself. Let $\varepsilon_{0,k}^*$ be an iid copy of $\varepsilon$ and let $\varepsilon_{s,k}^* =\varepsilon_{s,k}$ for $s\neq 0$, so that the only difference between $\xi_{s,k}$ and its coupled version $\xi_{s,k}^* = \xi_k(\varepsilon^*_{s-h,k} \colon h\in\mathbb{Z}^d)$ is the input at coordinate~$0$. \citet{machkouriaetal2013}\ call a random field \emph{$p$-stable} if $\sum_{s\in\mathbb{Z}^d}\Vert \xi_{s,k} - \xi_{s,k}^* \Vert_p < \infty$. Stability implies summability of the auto-covariance function of $(\xi_{s,k})_{s\in\mathbb{Z}^d}$ whenever $p\geq 2$ and is therefore a ``short-range'' dependence condition. If, e.g., the process is linear such that $\xi_{s,k} = \sum_{h\in\mathbb{Z}^d} \alpha_{h} \varepsilon_{s-h, k}$ and $\ev |\varepsilon|^p < \infty$, then $\Vert \xi_{s,k} - \xi_{s,k}^* \Vert_p$ = $|\alpha_s| \Vert \varepsilon_{0,k} - \varepsilon^*_{0,k} \Vert_p$ and $p$-stability is equivalent to $\sum_{h\in\mathbb{Z}^d} |\alpha_h| < \infty$. 

The following result is an immediate consequence of the moment bound developed by \citet{machkouriaetal2013}. It relates moments of weighted sums of vectors at arbitrary coordinates in the field to the stability property. It is useful in the standard case where the leading term $Z_{n,k}$ in Assumption \ref{as:decom} is linear. Recall that $\xi_{s,k}\in\mathbb{R}^{d_k}$.
\begin{proposition}[\citealp{machkouriaetal2013}]\phantomsection 
\label{p:momentbound}
Let $S$ be a finite subset of $\mathbb{Z}^d$ and let $(a_s)_{s\in S}$ be real numbers. If $\ev|\xi_{0,k}|^p<\infty$ for some $p\geq 2$, then \[\biggl\Vert\sum_{s\in S} a_s (\xi_{s,k} - \ev \xi_{s,k})\biggr\Vert_p \leq \sqrt[\uproot{2}2p]{d_k}\biggl(2p \sum_{s\in S} a_j^2\biggr)^{1/2} \sum_{s\in\mathbb{Z}^d}\Vert \xi_{s,k} - \xi_{s,k}^* \Vert_p.\] 
\end{proposition}

In the following, if $\xi_{s,k}$ is $p$-stable and has a subvector $V_{s,k}$, then I denote the corresponding subvector of the perturbed version $\xi_{s,k}^*$ by $V_{s,k}^*$. Clearly, if $\xi_{s,k}$ is $p$-stable, then $V_{s,k}$ must be $p$-stable as well. I repeatedly use this result and the proposition above in the next two examples.

\begin{example}[Regression with cluster-level treatment, continued]\phantomsection 
\label{ex:clusterreg3} Write the regression model \eqref{eq:regression} from Examples  \ref{ex:clusterreg} and \ref{ex:clusterreg2} as 
\begin{equation*}
Y_{i,k}  = \theta_{1\{k \leq q_{1,n}\}} + \eta_k' X_{i,k} + U_{i,k}.
\end{equation*} 
Let $\tilde{X}_{i,k} = (1, X_{i,k}')'$. As before, view each cluster as a separate regression and take $\smash{\hat{\theta}_{n,k}}$ to be the first entry of the least squares estimate \[\Biggl(\sum_{i=1}^{m_{n,k}} \tilde{X}_{i,k} \tilde{X}_{i,k}'\Biggr)^{-1}\sum_{i=1}^{m_{n,k}} \tilde{X}_{i,k}Y_{i,k}\] of $(\theta_{1\{k \leq q_{1,n}\}}, \eta_k')'$. If $\Sigma_{n,k} = m_{n,k}^{-1} \sum_{i=1}^{m_{n,k}} \tilde{X}_{i,k} \tilde{X}_{i,k}'$ converges in probability to a positive definite limit $\Sigma_k$, then $\sqrt{n/q_n}(\hat{\theta}_{n,k} - \theta_{1\{k \leq q_{1,n}\}})$ can be represented as
\begin{align*} 
e_1'\Sigma_k^{-1} \Bigl(\frac{m_{n,k}q_n }{n}\Bigr)^{1/2} m_{n,k}^{-1/2} \sum_{i=1}^{m_{n,k}} \tilde{X}_{i,k} U_{i,k} +e_1'(\Sigma_{n,k}^{-1} -\Sigma_k^{-1}) \Bigl(\frac{m_{n,k}q_n}{n}\Bigr)^{1/2} m_{n,k}^{-1/2} \sum_{i=1}^{m_{n,k}} \tilde{X}_{i,k} U_{i,k}
\end{align*}
where $e_1$ is a conformable vector with first element equal to $1$ and $0$ otherwise. Denote the first term by $Z_{n,k}$ and the second term by $R_{n,k}$. This is the decomposition required for Assumption \ref{as:decom}. Suppose $m_{n,k}q_n/n$ converges to a nonzero constant for every $k$ and $\sup_{n,k} m_{n,k}q_n/n < \infty$ so that all clusters provide similarly good estimates; this holds, for example, if the cluster sizes are constant multiples of one another.

I now discuss conditions (i)-(vi) of Assumption \ref{as:decom} in the present context. Condition (ii) is satisfied because $U_k$ is assumed to have conditional mean zero in Example \ref{ex:clusterreg}. For (iii), suppose that for each $k$, the $(X_{i,k}', U_{i,k})'$ come from a $2p$-stable random field in the sense that each $i$ is associated with a possibly unknown coordinate $s$ in the field. Then $\tilde{X}_{i,k}U_{i,k}$ is also part of a stationary random field and an application of the Cauchy-Schwarz inequality after adding and subtracting yields
\[\Vert \tilde{X}_{i,k}U_{i,k} - \tilde{X}_{i,k}^*U_{i,k}^* \Vert_p \leq \Vert U^*_{i,k} \Vert_{2p}\Vert \tilde{X}_{i,k} - \tilde{X}_{i,k}^* \Vert_{2p} + \Vert \tilde{X}_{1,k} \Vert_{2p}\Vert U_{i,k} - U_{i,k}^*\Vert_{2p},\]
where $\Vert U^*_{i,k} \Vert_{2p} = \Vert U_{i,k} \Vert_{2p} = \Vert U_{1,k} \Vert_{2p}$ because $U^*_{i,k}$ and $U_{i,k}$ are identically distributed and $U_{i,k}$ is stationary in $i$. Conclude that if $(X_{i,k}', U_{i,k})'$ is $2p$-stable, then $\tilde{X}_{i,k}U_{i,k}$ is $p$-stable. Proposition \ref{p:momentbound} now implies that there is a constant $C_{p,k}$ depending on only $p$ and $k$ with
\begin{equation}\label{eq:zbound}
\Vert Z_{n,k} \Vert_{p} \lesssim |e_1'\Sigma_k^{-1}|\biggl\Vert m_{n,k}^{-1/2} \sum_{i=1}^{m_{n,k}} \tilde{X}_{i,k} U_{i,k} \biggr\Vert_p \leq C_{p,k}.
\end{equation}
If the entire set of clusters consists of $2p$-stable random fields and the eigenvalues of $\Sigma_k$ are uniformly bounded away from zero, then $\sup_k C_{p,k} < \infty$ and condition (iii) is satisfied. Condition (iv) follows because the random fields are independent and (v) has to be assumed. 

\phantomsection\label{rev:remainderrates}%
For (vi), let $\Sigma_k = \ev \tilde{X}_{1,k}\tilde{X}_{1,k}'$ and apply first the Chebyshev inequality and then Proposition~\ref{p:momentbound} to bound $\prob(|\Sigma_{n,k} - \Sigma_k| > \epsilon)$ by $\epsilon^{-2}$ times 
\begin{equation}\label{eq:mmatrixbound}
 \Vert \Sigma_{n,k} - \Sigma_k \Vert_2^2 = \biggl\Vert m_{n,k}^{-1} \sum_{i=1}^{m_{n,k}} (\tilde{X}_{i,k} - \ev \tilde{X}_{i,k}) \biggr\Vert_4^4 \lesssim m_{n,k}^{-2}.
\end{equation}
To transform this into a statement about inverses, write $\Sigma_{n,k}^{-1} -\Sigma_k^{-1} = \Sigma_{n,k}^{-1} (\Sigma_k - \Sigma_{n,k})\Sigma_k^{-1}$ to obtain the bound $|\Sigma_{n,k}^{-1} -\Sigma_k^{-1}|\leq |\Sigma_{n,k}^{-1}||\Sigma_{k}^{-1}||\Sigma_{n,k} - \Sigma_k|$. Because the eigenvalues of $\Sigma_{k}^{-1}$ are uniformly bounded away from zero, there is a constant $M$ such that $\sup_k |\Sigma_{k}^{-1}| \leq M$. Define the events $S_{n,k} = \{ |\Sigma_{n,k}^{-1}|\leq M \}$ and $S_n = \cap_k S_{n,k}$. Apply the Markov inequality, then the H\"older inequality with exponents $p/(p-1)$ and $p$, and finally \eqref{eq:zbound} to see that $\prob ( \sum_{k=1}^{q_n} |R_{n,k}| > \epsilon \sqrt{q_n}, S_n)$ is bounded above by a constant multiple of \[  q_n^{-1/2}\sum_{k=1}^{q_n} \ev |\Sigma_{n,k} - \Sigma_k| \biggl\vert m_{n,k}^{-1/2} \sum_{i=1}^{m_{n,k}} \tilde{X}_{i,k} U_{i,k} \biggr\vert \lesssim q_n^{-1/2}\sum_{k=1}^{q_n} \Vert \Sigma_{n,k} - \Sigma_k\Vert_{p/(p-1)}. \] Because $\Vert \Sigma_{n,k} - \Sigma_k\Vert_{p/(p-1)}\leq \Vert \Sigma_{n,k} - \Sigma_k\Vert_2\lesssim m_{n,k}^{-1}$, the right-hand side of the display is $O(\sqrt{q_n}/\inf_k m_{n,k})$. Now consider $\prob(S_n^c)\leq \sum_{k=1}^{q_n} \prob(|\Sigma_{n,k}^{-1}| >  M)$. Denote by $\lambda_{\min}(A)$ and $\lambda_{\max}(B)$ the smallest and largest eigenvalues of symmetric matrices $A$ and $B$. Let $r = \sup_k \sqrt{\rank \Sigma_k}$. Then $\prob(|\Sigma_{n,k}^{-1}| >  M) \leq \prob(\lambda_{\min}(\Sigma_{n,k}) < r/M)$ because the matrix Euclidean norm is also the 2-Schatten norm. By the Weyl inequality $|\lambda_{\min}(A) - \lambda_{\min}(B)|\leq \lambda^{1/2}_{\max}((A-B)'(A-B))\leq | A-B |$, this is bounded by \[ \prob\bigl(\lambda_{\min}(\Sigma_{k}) < r/M + | \Sigma_{n,k} - \Sigma_k |\bigr) \leq 1\{ \lambda_{\min}(\Sigma_{k}) < r/M + \delta \} + \prob\bigl(| \Sigma_{n,k} - \Sigma_k | > \delta \bigr) \] for $\delta > 0$.
Choose $M$ large and $\delta$ small enough so that $\inf_k\lambda_{\min}(\Sigma_{k}) \geq r/M + \delta$. Now use the Chebyshev inequality and \eqref{eq:mmatrixbound} to conclude \[ \prob(S_n^c) \leq \sum_{k=1}^{q_n} \prob\bigl(| \Sigma_{n,k} - \Sigma_k | > \delta\bigr) \lesssim \frac{q_n}{\inf_k m_{n,k}^2}. \]
If follows that $\sum_{k=1}^{q_n} |R_{n, k}| = \op(\sqrt{q_n})$ as long as $\sqrt{q_n}/\inf_k m_{n,k}\to 0$, which allows for a wide range of sequences $q_n\to\infty$.
\sqed
\end{example}

\begin{example}[Binary choice with cluster-level treatment]\label{ex:binreg} Suppose the model in the preceding example is the latent model in a binary choice framework where $U_{i,k}$ is independent of treatment assignment and $X_{i,k}$, $U_{i,k}$ has a known, symmetric distribution function $F$, and $\eta_k \equiv \eta_0$. The treatment effect of interest is $F(\theta_1 + \eta_0' x) - F(\theta_0 + \eta_0' x)$ for some $x$ so that $H_0\colon \theta_1 = \theta_0$ is the appropriate hypothesis to test. Assume that, for each $k$, the $(Y_{i,k}, X_{i,k}')'$ come from a stationary random field. Only $1\{Y_{i,k} > 0\}$, $X_{i,k}$, and treatment assignment are observed. Let $\psi_{\theta, \eta}(y, x) = (1, x')'(1\{ y > 0 \} - F(\theta + \eta'x) )$ and suppose the moment condition $\ev \psi_{\theta_{1\{k \leq q_{1,n}\}}, \eta_0}(Y_{1,k}, X_{1,k}) = 0$ holds for every $k$. The corresponding $Z$-estimates $(\hat{\theta}_{n,k}, \hat{\eta}_{n,k}')'$ for the $k$-th cluster are zeros of
\[ (\theta, \eta')'\mapsto \Psi_{n, k}(\theta, \eta) = m_{n,k}^{-1}\sum_{i=1}^{m_{n,k}} \psi_{\theta, \eta}(Y_{i,k}, X_{i,k}). \] 
Denote the derivative of $\Psi_{n, k}$ with respect to $(\theta, \eta')$ by $\dot{\Psi}_{n,k}$ and let $\dot{\Psi}_k = \ev \dot{\Psi}_{n,k}$. If $\dot{\Psi}_k(\theta_{1\{k \leq q_{1,n}\}}, \eta_0)$ is non-singular, then standard arguments for $Z$-estimators \citep[see, e.g.,][]{vandervaart1998} suggest that one can take $Z_{n,k}$ as the first term on the right in
\[ \sqrt{n/q_n}(\hat{\theta}_{n,k} - \theta_{1\{k \leq q_{1,n}\}}) = e_1' \dot{\Psi}_k(\theta_{1\{k \leq q_{1,n}\}}, \eta_0)^{-1} \sqrt{n/q_n} \Psi_{n, k}(\theta_{1\{k \leq q_{1,n}\}}, \eta_0) + R_{n,k}, \]
where $R_{n,k}$ is defined as the difference between the left-hand side and $Z_{n,k}$.

Let $\tilde{X}_{i,k} = (1, X_{i,k}')'$ and assume that $F$ is Lipschitz. Using arguments as in the preceding example, apply the Cauchy-Schwarz inequality after repeated adding and subtracting to bound $\Vert \psi_{\theta_j, \eta}(Y_{i,k}, X_{i,k}) - \psi_{\theta_j, \eta}(Y_{i,k}^*, X_{i,k}^*) \Vert_p$ by a constant multiple of
\[ \Vert \tilde{X}_{i,k} - \tilde{X}_{i,k}^* \Vert_{p} + \Vert\tilde{X}_{i,k} - \tilde{X}_{i,k}^* \Vert_{2p} + \Vert 1\{Y_{i,k} \leq 0\} -  1\{Y_{i,k}^* \leq 0\} \Vert_{2p}.\]
Because $|1\{a<0\} - 1\{b < 0\}|\leq 1\{|a| < |a-b|\}$ for $a,b\in\mathbb{R}$, the last term in the display is at most $\Vert 1\{|Y_{i,k}| \leq |Y_{i,k}-Y_{i,k}^*|\} \Vert_{2p} \leq\Vert 1\{|Y_{i,k}| \leq \epsilon \}\Vert_{2p} + \Vert 1\{|Y_{i,k}-Y_{i,k}^*| > \epsilon \} \Vert_{2p}$ for every $\epsilon > 0$. Use Lipschitz continuity of $F$ and the Chebyshev inequality to bound this by a constant multiple of $\epsilon^{1/(2p)} + \epsilon^{-1}\Vert Y_{i,k}-Y_{i,k}^* \Vert_{2p}$. Choose $\epsilon = \Vert Y_{i,k}-Y_{i,k}^* \Vert_{2p}^{2p/(2p+1)}$. Then the preceding display does not exceed a constant multiple of 
\[ \Vert \tilde{X}_{i,k} - \tilde{X}_{i,k}^* \Vert_{p} + \Vert\tilde{X}_{i,k} - \tilde{X}_{i,k}^* \Vert_{2p} + \Vert Y_{i,k}-Y_{i,k}^* \Vert_{2p}^{1/(2p+1)}.\] Hence, $\psi_{\theta_j, \eta}(Y_{i,k}, X_{i,k})$ is $p$-stable as long as $X_{i,k}$ is $2p$-stable and $Y_{i,k}$ satisfies the $2p$-stability condition with $\Vert\cdot\Vert_{2p}$ strengthened to $\Vert\cdot\Vert_{2p}^{1/(2p+1)}$. (The last condition can be weakened by imposing a stronger stability condition on $X_{i,k}$.) Conclude from Proposition \ref{p:momentbound} that Assumption \ref{as:decom}(iii) is satisfied. Condition (vi) can be established with similar arguments and Proposition \ref{p:uniformremainder}. The remaining conditions follow as in Example \ref{ex:clusterreg3}.\sqed
\end{example}

\section{Numerical results}\label{s:montecarlo}
This section presents several Monte Carlo experiments to illustrate the small-sample properties of the placebo test in comparison to other methods of inference. I discuss linear regression with cluster-level treatment (Example \ref{ex:clusterreg4}), probit regression with cluster-level treatment (Example \ref{ex:binreg2}), and difference-in-differences estimation using placebo interventions in a sample from the Current Population Survey (Example \ref{ex:diffindiff3}). Finally, I use the placebo test to analyze data from an experiment on infinitely repeated games (Example \ref{ex:dalbofrechette}).

\begin{example}[Regression with cluster-level treatment, continued]\label{ex:clusterreg4} 
The regression of interest is a simplified version of \eqref{eq:regression} with $\eta_k\equiv \eta$,
\begin{equation}\label{eq:ex1eq}
	Y_{i,k} = \theta_0 + \beta D_{k} + \eta' X_{i,k} + U_{i,k}.
\end{equation}
For the experiment, the errors $U_{i,k}$ have a ``circular'' dependence structure defined by
\begin{equation}\label{eq:circular}
U_{i,k} = \frac{1}{h+1}\sum_{j = i}^{i+h } \tilde{U}_{c(j),k}, \qquad 1\leq i\leq m_{n,k}\text{ and }c(j) = j~(\bmod{m_{n,k}}),
\end{equation}
where $h \in \{0, 1, \dots, m_{n,k}-1 \} $ and, for each $k$ and $i$, $\tilde{U}_{i,k}$ is independently drawn from $\mathrm{N}(0, 1)$ if $k\leq q_{1,n}$ and $\mathrm{N}(0, 2)$ if $k > q_{1,n}$. This ensures that the $U_{i,k}$ are non-identically distributed and independent across clusters but identically $h$-dependent within clusters. Each element of the $5$-vector of covariates $X_{i,k}$ is also $h$-dependent and is generated as in the preceding display but with independent draws from $\mathrm{N}(0, 1)$ if $k\leq q_{1,n}$ and $\chi^2_2-2$ if $k > q_{1,n}$. The default degree of dependence is $h=10$ and the cluster sizes $m_{n,k}$ are uniformly distributed on $\{ 15,\dots, 25 \}$. For all simulations below I set $\theta_0 = 0$ and $\eta = (1, 1, 1, 1, 1)'$.

In this experiment, I test the hypothesis $H_0\colon \beta = 0$ against the one-sided alternative $H_1\colon \beta > 0$ for different values of the treatment effect $\beta$, the number of treated and untreated clusters $q_{1,n}$ and $q_{0,n}$, and the degree of dependence $h$. I consider (i) placebo inference as in Algorithm \ref{al:placebo}, (ii) two-sample inference with the \citet{ibragimovmueller2016} $t$ test, (iii) randomized and non-randomized inference with the permutation $t$ test of \citet{canayetal2014}, (iv) cluster-robust wild bootstrap inference as described in \citet{cameronetal2008} and \citet{webb2013}, and (v) inference based on the cluster-robust regression $t$ statistic of \citet{besteretal2014}. Several other methods for cluster-robust inference exist; see, e.g., \citet{mackinnonwebb2017} for one such approach and an overview of the recent literature. They are omitted because these methods tend to be designed for very specific situations.

I now give brief descriptions of methods (i)-(v). For (i), I use the unadjusted critical values $\mean{c}_{n,\alpha}$ defined below Algorithm \ref{al:placebo} whenever $q_{1,n}=q_{0,n}$ and use Algorithm \ref{al:placebo} otherwise; the results with adjusted critical values at $q_{1,n}=q_{0,n}$ were very similar and are therefore not reported. 

For (ii), the \citetalias{ibragimovmueller2016} test compares the absolute value of \[ \frac{\mean{T}_n(\hat{\theta}_n)}{\hat{S}_n(\hat{\theta}_n)} \] to the $1-\alpha$ quantile of a $t$ distribution with $\min\{q_{1,n}, q_{0,n} \} - 1$ degrees of freedom. They show that this is a valid two-sided test for a range of fixed $q_{1,n}$ and $q_{0,n}$ if the $\sqrt{n}(\hat{\theta}_{n,k}-\theta_0)$ are asymptotically independent and normal, where the range is determined by $\alpha$. Simulations suggest that the statistic in the preceding display can be used for a one-sided test, which is what I do here for ease of comparison. 

For (iii), the \citetalias{canayetal2014} test only applies if $q_{1,n}=q_{0,n}$ and requires the researcher to match a treated cluster to a control cluster to identify $\beta$. Because both the group of treated clusters and the group of control clusters are relatively homogenous, I match treated and control clusters at random and estimate $\hat{\beta}_{n,k}$ by least squares in \eqref{eq:regression} for each of the $q_{1,n}$ matched pairs. Let $\smean{\beta}_n = q_{1,n}^{-1} \sum_{k=1}^{q_{1,n}}\hat{\beta}_{n,k}$. The \citetalias{canayetal2014} test compares \[ \frac{\smean{\beta}_n}{\sqrt{\sum_{k=1}^{q_{1,n}}(\hat{\beta}_{n,k}- \smean{\beta}_n)^2}} \] to the $1-\alpha$ quantile of the permutation distribution obtained by recomputing this statistic for all possible sign changes of the elements of $(\hat{\beta}_{n,1}, \dots, \hat{\beta}_{n,q_{1,n}})$. They show that their test is valid for fixed $q_{1,n}=q_{0,n}$ if the asymptotic distribution of $\sqrt{n}(\hat{\beta}_{n,1}, \dots, \hat{\beta}_{n,q_{1,n}})$ is asymptotically symmetric and some technical regularity conditions are satisfied. As mentioned in the Remarks below Algorithm \ref{al:placebo}, this test can be randomized. It should also be noted that the experimental design likey overstates the performance of the \citetalias{canayetal2014} test because there are $q_{1,n}!$ ways of matching pairs and hence there are equally many ways of computing the test statistic in the preceding display. In practice one may need to use a multiple testing adjustment to account for this indeterminacy, which can lead to significantly lower rejection rates. 

For (iv), I estimate $\beta$ via least squares in the pooled sample and standardize this estimate with the usual cluster-robust covariance matrix with degrees-of-freedom adjustment $(n-1)q_n/((n-d)(q_n-1))$, where $d$ is the number of controls in the pooled regression. I then compare this estimate to the bootstrap distribution of the $t$ statistic obtained from 199 repetitions of the cluster-robust version of the wild bootstrap using the \citet{webb2013} 6-point distribution and with the null hypothesis imposed on the data. This procedure is outlined in detail in \citet{cameronetal2008}. 

For (v), I compare the statistic from (iv) to the $1-\alpha$ quantile of $t$ distribution with $q_n-1$ degrees of freedom. \citet{besteretal2014} show that this test is valid for certain ranges of fixed $q_n$, where the range depends on $\alpha$, if the distribution of the covariates is very similar across clusters and other technical regularity conditions are satisfied.

\begin{figure}[t]
\centering
\includegraphics[width=.96\textwidth]{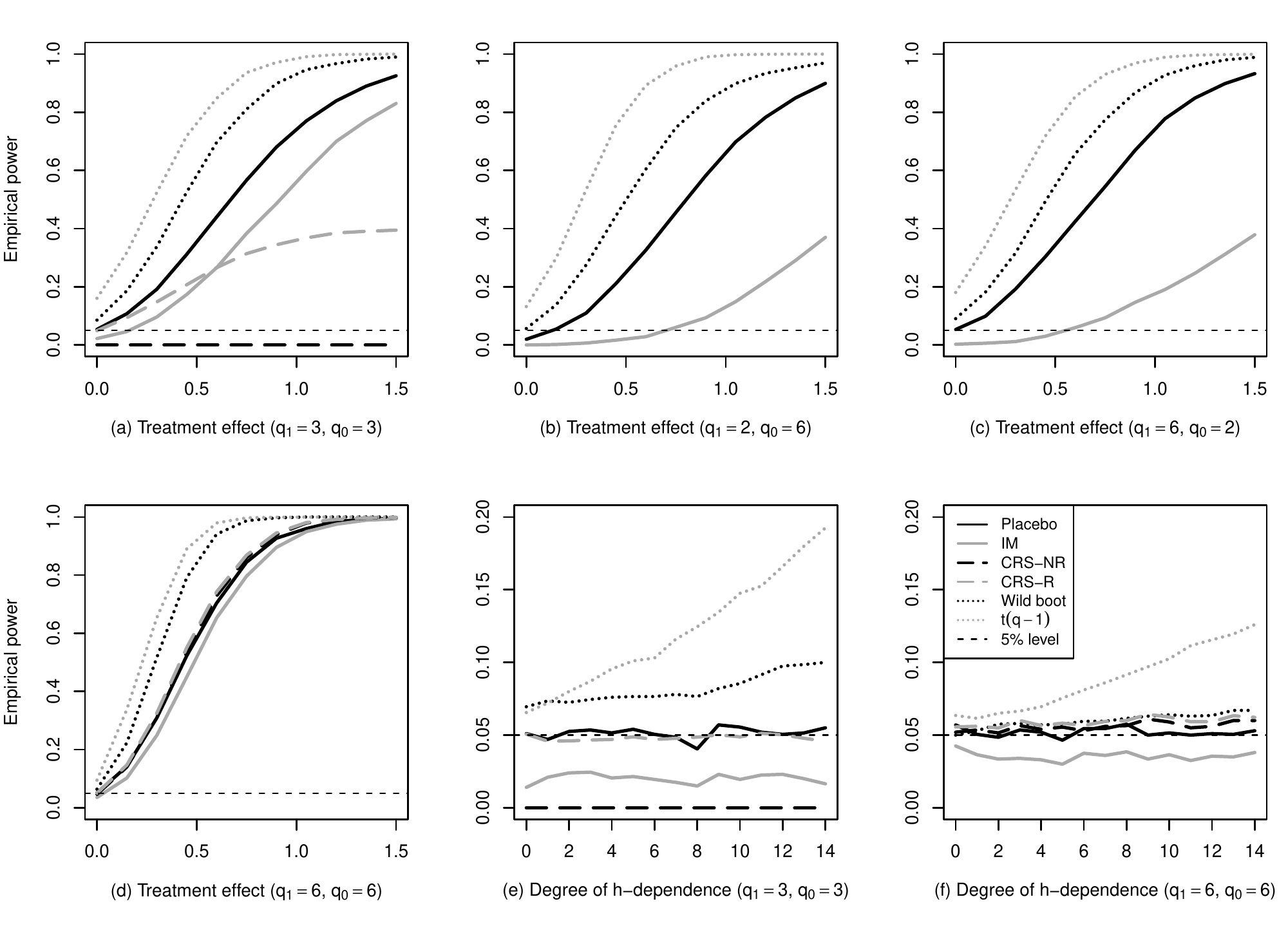}
\caption{Empirical rejection frequencies of the hypothesis $H_0\colon \beta = 0$ in Example \ref{ex:clusterreg4} using the placebo test (solid black lines), \citetalias{ibragimovmueller2016} test (solid grey), \citetalias{canayetal2014} non-randomized (long-dashed black), and randomized permutation test (long-dashed grey), wild bootstrap (dotted black), and $t(q_n-1)$ critical values (dotted grey) at the 5\% level (short-dashed) as a function of the size of the treatment effect for (a)~$q_{1,n} = q_{0,n} = 3$, (b) $q_{1,n} = 2$, $q_{0,n} = 6$, (c) $q_{1,n} = 6$, $q_{0,n} = 2$, (d) $q_{1,n} = q_{0,n} = 6$, and as a function of the degree of dependence $h$ under $H_0$ for (e) $q_{1,n} = q_{0,n} = 3$ and (f) $q_{1,n} = q_{0,n} = 6$. The \citetalias{canayetal2014} test does not apply to (b) and (c).} \label{fig:mc_ex11}
\end{figure}

Panels (a)-(d) of Figure \ref{fig:mc_ex11} show the empirical rejection frequencies of the hypothesis $H_0\colon \beta = 0$ for methods (i)-(v) at the 5\% level (short-dashed line) as a function of $\beta\in \{0, .15, .3,\dots, 1.5\}$ for (a) $q_{1,n} = q_{0, n} = 3$, (b) $q_{1,n} = 2$, $q_{0, n} = 6$, (c) $q_{1,n} = 6$, $q_{0, n} = 2$, and (d) $q_{1,n} = q_{0, n} = 6$. Panels (e)-(f) plot the  empirical rejection frequencies of the correct hypothesis $H_0\colon \beta = 0$ for the same methods as a function of $h\in\{0, 1,\dots, 14\}$ for (e) $q_{1,n} = q_{0, n} = 3$ and (f) $q_{1,n} = q_{0, n} = 6$. Each horizontal coordinate was computed from 2,000 simulations and all five methods were faced with the same data.

\phantomsection\label{rev:morenumbers}%
As can be seen in panel (a), the placebo test (solid black line) performed well with three treated and three untreated clusters. It rejected in 5.35\% of all cases under the null and its power increased quickly as a function of $\beta$. In comparison, the \citetalias{ibragimovmueller2016} test (solid grey) was quite conservative (2.2\%) and had lower power than the placebo test. As pointed out in the Remarks below Algorithm \ref{al:placebo}, the \citetalias{canayetal2014} non-randomized (long-dashed black) test has power equal to zero for all values of $\beta$ at the 5\% level if less than 5 treated clusters are available. The empirically less relevant randomized version of the permutation test (long-dashed grey) was nearly exact at $\beta = 0$ but had lower power than the placebo test. The wild bootstrap rejected a true null hypotheses in 8.55\% of all cases. The $t$ test over-rejected more severely with a rate of 16.05\%. Because of this size distortion, their rejection frequencies for nonzero values of $\beta$ should not be interpreted as estimates of their power. In (b), the placebo test behaved more conservatively (1.65\%) when $q_{1,n} = 2$ and $q_{0, n} = 6$ but far outperformed the \citetalias{ibragimovmueller2016} test for all values of $\beta$. The wild bootstrap clearly dominated all other methods of inference in this instance. However, as can be seen in (c), this appears to be an artifact of the experimental design. At $q_{1,n} = 6$ and $q_{0, n} = 2$, the wild bootstrap now over-rejected substantially with an empirical rejection frequency of 9.05\% under the null. The power of the \citetalias{ibragimovmueller2016} test increased slightly, whereas the placebo test was essentially at the nominal level (5.30\%) and had high power. The $t$ test was again less reliable than other tests both in (b) and (c). Once I increased the number of clusters to $q_{1,n} = q_{0, n} = 6$ in (d), the size distortion in the wild bootstrap disappeared and it provided a test with high power. The $t$ test rejected in 9.50\% of all cases. The placebo, \citetalias{canayetal2014}, and \citetalias{ibragimovmueller2016} tests gave similar results, with the \citetalias{ibragimovmueller2016} rejecting the fewest hypotheses. The consistent over-rejection of the wild bootstrap in this example can also be seen in (e), where the rejection frequency of the wild bootstrap increased with the degree of dependence. The over-rejection of the $t$ test became much more severe with $h$. In contrast, the placebo test and the randomized \citetalias{canayetal2014} test were essentially at the nominal level for all values of $h$. The \citetalias{ibragimovmueller2016} test was very conservative in all cases but this improved when I increased the number of clusters to $q_{1,n} = q_{0, n} = 6$ in (f). Here, the $t$ test over-rejected but was less affected by the degree of dependence. All other methods were nearly exact.

\phantomsection\label{rev:wildboothomo}%
I also experimented with the distribution of the errors and covariates (not shown). As can be expected from the Monte Carlo simulations of \citet{mackinnonwebb2014,mackinnonwebb2017}, the performance of the wild bootstrap improved considerably if the clusters were more homogenous. Similarly, comparing the cluster-robust least squares $t$ statistic to the $t(q_n-1)$ distribution worked much better once the homogeneity conditions of \citet{besteretal2014} were satisfied and the number of clusters was larger. The other methods where less affected by changes in the distribution of the errors and covariates.
\sqed
\end{example}

\begin{example}[Binary choice with cluster-level treatment, continued]\label{ex:binreg2}
This example uses \eqref{eq:ex1eq} as the latent model in a probit regression. The data are as in Example \ref{ex:clusterreg4} except that the $\tilde{U}_{i,k}$ variables used to generate the $U_{i,k}$ in \eqref{eq:circular} are now standard normal for all $k$. In addition, the cluster sizes are uniformly distributed on $\{350,\dots, 500 \}$ to facilitate computation. 
\begin{figure}[t]
\centering
\includegraphics[width=.65\textwidth]{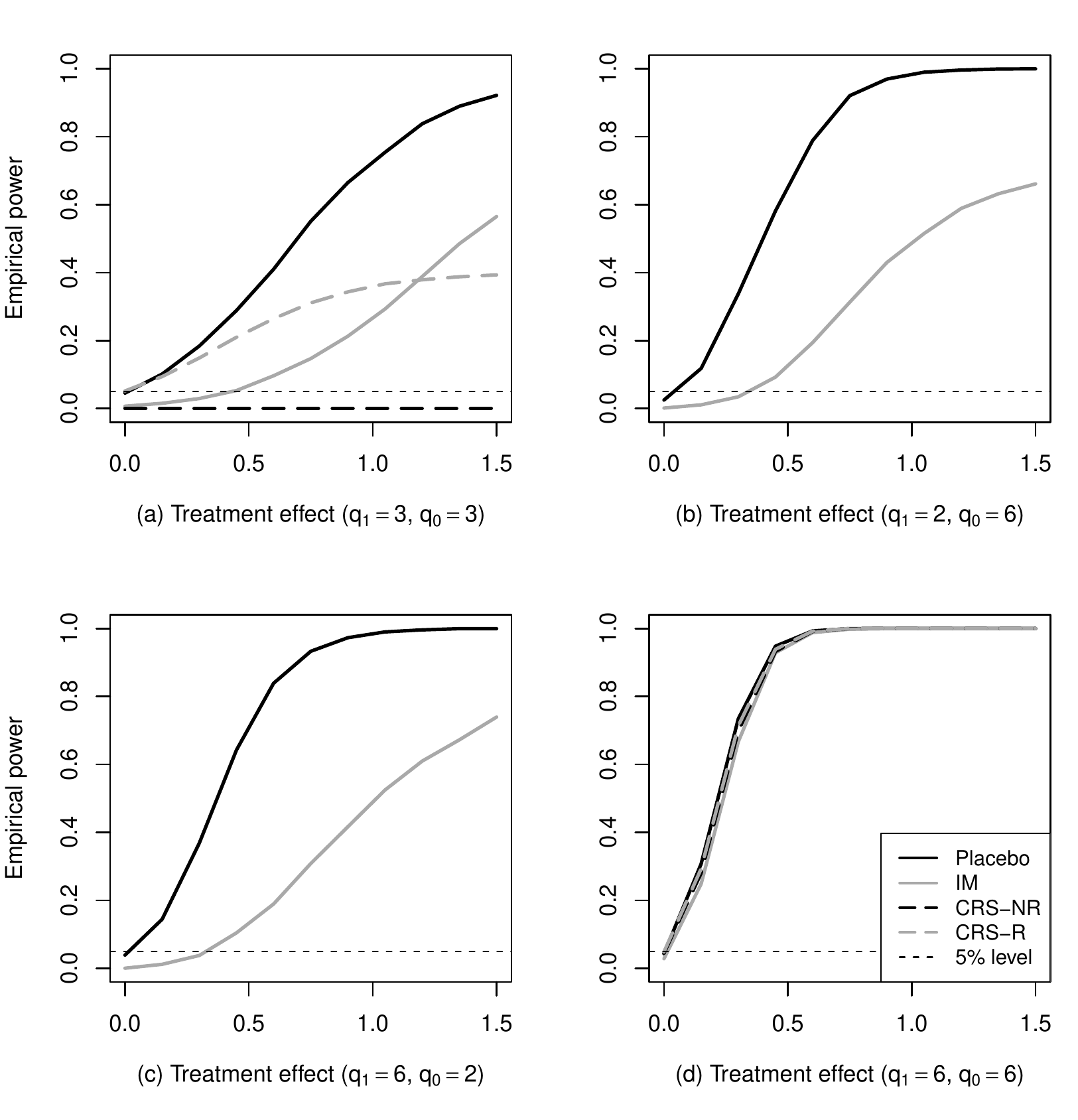}
\caption{Empirical rejection frequencies of the hypothesis $H_0\colon \beta = 0$ in Example \ref{ex:binreg2} with labels as in Figure \ref{fig:mc_ex11} as a function of the size of the treatment effect for (a)~$q_{1,n} = q_{0,n} = 3$, (b) $q_{1,n} = 2$, $q_{0,n} = 6$, (c) $q_{1,n} = 6$, $q_{0,n} = 2$, and (d) $q_{1,n} = q_{0,n} = 6$.} \label{fig:mc_ex21}
\end{figure}
Figure \ref{fig:mc_ex21} repeats the experiments in Figure \ref{fig:mc_ex11}(a)-(d) in the probit case for $\beta\in \{0, .05, .1,\dots, 1.5\}$ with the placebo test, the \citetalias{ibragimovmueller2016} test, and the \citetalias{canayetal2014} test. The wild bootstrap and the \citetalias{besteretal2014} $t$ test are primarily designed for linear models and are therefore not included here. The placebo test is again seen to provide a test near nominal level at (a) $q_{1,n} = q_{0, n} = 3$, (b) $q_{1,n} = 2$, $q_{0, n} = 6$, and (c) $q_{1,n} = 6$, $q_{0, n} = 2$, whereas the \citetalias{ibragimovmueller2016} test behaved considerably more conservatively. At $q_{1,n} = q_{0, n} = 6$, all methods controlled size and had high power. \sqed
\end{example}

\begin{example}[Difference in differences, continued]\label{ex:diffindiff3} For this experiment, I follow \citet{bertrandetal2004} and use data from the Merged Outgoing Rotation Group of the Current Population Survey (CPS) of women in their fourth interview month.  They consider women between the ages 25 and 50 with strictly positive earnings. I extract data on weekly earnings, age, education, and state of residence. The outcome of interest is log weekly earnings in a difference-in-differences regression as in \eqref{eq:diffindiff}. The controls in this regression are a quadratic polynomial in age, dummies for education (less than high school, high school, some college, and college or more), and state and year fixed effects. I focus on the years 1989-1991. This sample contains 86,006 observations with an average of approximately 573 women per state-year. For the experiment, I sampled $q_n$ states at random from the 50 states and each state over time is viewed as a single cluster. I then assigned treatment in 1990 to half of these clusters even though no treatment was received in the actual data at this time. 

\begin{figure}[t]
\centering
\includegraphics[width=.96\textwidth]{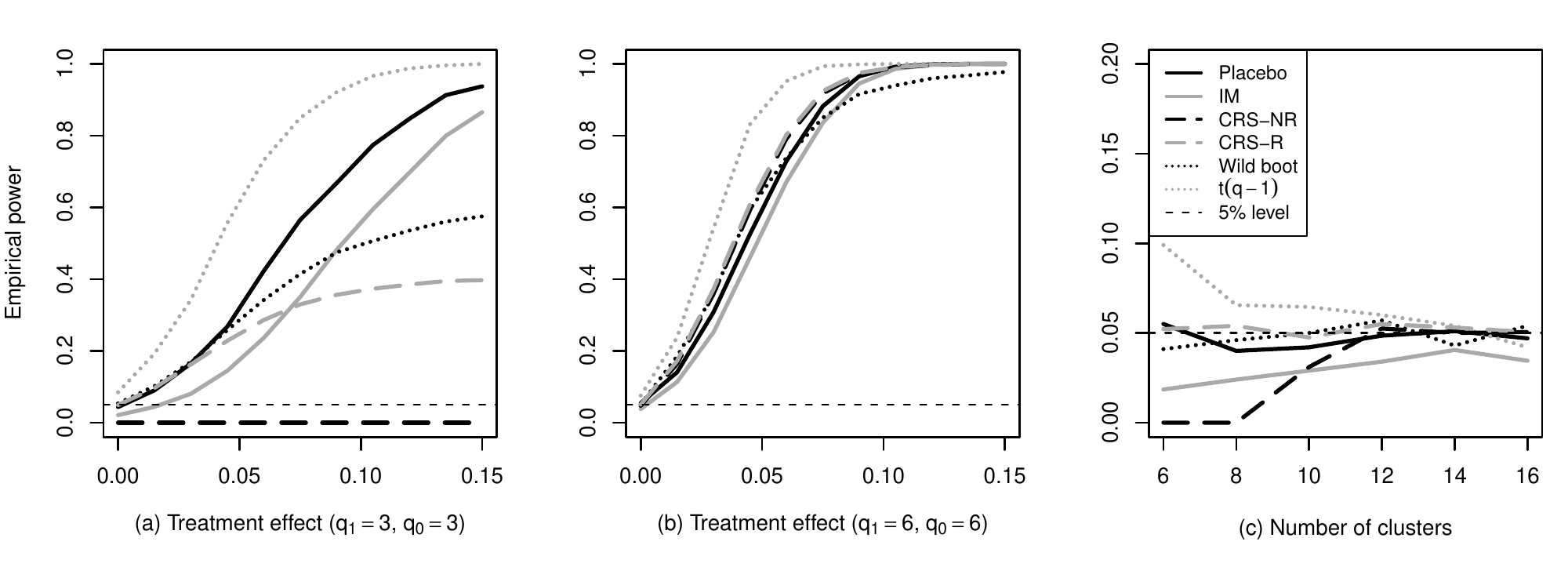}
\caption{Empirical rejection frequencies of the hypothesis $H_0\colon \beta = 0$ in Example \ref{ex:diffindiff3} with labels as in Figure \ref{fig:mc_ex11} as a function of the size of the treatment effect for (a) $q_{1,n} = q_{0, n} = 3$ and (b) $q_{1,n} = q_{0, n} = 6$, and (c) as a function of the number of clusters under $H_0$.} \label{fig:mc_ex31}
\end{figure}

Panels (a)-(b) of Figure \ref{fig:mc_ex31} show the empirical rejection frequencies of the hypothesis $H_0\colon \beta = 0$ for the methods of inference used in Example \ref{ex:clusterreg4} at the 5\% level (short-dashed line) as a function of $\beta\in \{0, .015, .03,\dots, 0.15\}$ for (a) $q_{1,n} = q_{0, n} = 3$ and (b) $q_{1,n} = q_{0, n} = 6$. I generated the data for these two panels by adding $\beta$ to the outcomes in treated states for years 1990 and 1991. Panel (c) plots the empirical rejection frequencies of the correct hypothesis $H_0\colon \beta = 0$ for the same methods as a function of $q_n\in \{6, 8, \dots, 16 \}$ with $q_{1,n} = q_{0, n}$. Each horizontal coordinate was computed from 2,000 simulations and all five methods were faced with the same data. I matched the clusters for the \citetalias{canayetal2014} test by size. As can be seen in panel (a), the CPS data was considerably less challenging than the simulated data from Example~\ref{ex:clusterreg}. The placebo test, the wild bootstrap, and the randomized \citetalias{canayetal2014} test provided an essentially exact test, although the placebo test had significantly higher power for most values of $\beta$. The \citetalias{ibragimovmueller2016} test was conservative but surpassed the wild bootstrap and the \citetalias{canayetal2014} test in terms of power as $\beta$ increased. The $t$ test again over-rejected. All methods improved when I increased the number of clusters to 12 in (b). Notable here is the \citetalias{besteretal2014} $t$ test, which marginally over-rejected but provided a test with very high power. As panel (c) shows, all methods behaved quite similarly as the number of clusters increased further. 

I also experimented with different years in this data set (not shown), but the results were similar. Furthermore, deviating from more standard tools such as the wild bootstrap and the \citetalias{besteretal2014} $t$ test did not seem necessary in the CPS data if a moderate and relatively balanced number of treatment and control clusters was available.  \sqed
\end{example}

\begin{example}[Infinitely repeated games; \citealp{dalbofrechette2011}] \label{ex:dalbofrechette} In this empirical application, I follow \citet{ibragimovmueller2016} and reanalyze a lab experiment of \citet{dalbofrechette2011}, who measure cooperation between players in an infinitely repeated prisoner's dilemma with a continuation probability of either $\delta = .5$ or $\delta = .75$. Cooperation resulted in a payoff of $R = 32$, $40$, or $48$. Each $(\delta, R)$ combination was played in a separate experiment and three sessions per $(\delta, R)$ combination were held. A total of 37,042 games were completed with an average of approximately 2,058 games per session. I view each session as a cluster. 

Cooperation can be supported as a subgame-perfect equilibrium action for all $(\delta, R)$ combinations other than $(.5, 32)$. (See \citeauthor{dalbofrechette2011}'s discussion for a more precise statement.) These cooperative equilibria are Pareto efficient. The rate of cooperation $\theta(\delta, R)$ should therefore increase if either $\delta$ increases to $\delta'$ or $R$ increases to $R'$. It is then natural to test hypotheses of the form $H_0\colon \beta = 0$ against $H_1\colon \beta > 0$ for $\beta = \theta(\delta', R) - \theta(\delta, R)$ or $\beta = \theta(\delta, R') - \theta(\delta, R)$, where each increase of $R$ by eight for fixed $\delta$ or each increase in $\delta$ from $.5$ to $.75$ for fixed $R$ is viewed as a treatment. This leads to seven separate null hypotheses. For each hypothesis, three treated sessions and three control sessions are available, which makes the data amenable to the placebo test.
	
\begin{table}[t!]
\caption{Empirical frequency of cooperation in \citet[Table~3, bottom-right panel]{dalbofrechette2011} with tests of equal vs.\ more (``$<$'' in rows, ``$\wedge$'' in columns) cooperation between experimental setups using (a) cluster-robust standard errors and (b) Algorithm~\ref{al:placebo}.}\label{tab:dalbofrechette}
\begin{tabular}{cp{0cm}cccccp{0cm}ccccc}%
\hline
& 
& \multicolumn{5}{c}{(a) \citet{dalbofrechette2011}}
& 
& \multicolumn{5}{c}{(b) Placebo test}
\\ 
\cline{3-7}\cline{9-13}
$\delta\backslash R$
&
& 32 & & 40 & & 48 
&
& 32 & & 40 & & 48  
\\ 
\hline
$.50$ & & .0982 & $<^{**{\phantom{*}}}$ 	& .1798 & $<^{***}$ 				& .3529 &  
		& .0982 & $<^{*\phantom{**}}$ 	& .1798 & $<^{**{\phantom{*}}}$ 	& .3529  \\ 
      & & $\wedge^{**{\phantom{*}}}$  & & $\wedge^{***}$ 			  & & $\wedge^{***}$ 
      & & $\wedge^{*\phantom{**}}$ 	  & & $\wedge^{**{\phantom{*}}}$  & & $\wedge^{**{\phantom{*}}}$ \\ 
$.75$ &	& .2025 & $<^{***}$ 			& .5871 & $<^{*{\phantom{**}}}$ & .7642  
	  & & .2025 & $<^{**{\phantom{*}}}$ 	& .5871 & $<^{\phantom{***}}$ 	& .7642  \\  
\hline
\multicolumn{13}{l}{{\footnotesize\emph{Note:} * indicates significance at 10\%, ** at 5\%, *** at 1\%}}
\end{tabular}%
\end{table}

The two panels of Table \ref{tab:dalbofrechette} reproduce the estimates of $\theta(\delta, R)$ from \citet[Table~3, bottom-right panel]{dalbofrechette2011} but differ in the way inference was performed. For panel (a), I follow \citeauthor{dalbofrechette2011} use probit regressions with a dummy for one of the two $(\delta, R)$ combinations under consideration and compute cluster-robust standard errors. \citeauthor{dalbofrechette2011} report results from two-sided tests. I compute one-sided tests, which makes two of their results more significant (as indicated by the number of asterisks). In either case, the finding is that higher continuation probabilities or payoffs lead to significantly higher rates of cooperation at conventional levels of significance. For panel (b), I apply the placebo test after computing probit regressions on a constant within each of the six sessions under consideration. As can be seen, in each case the difference in cooperation rates loses one level of significance but \citeauthor{dalbofrechette2011}'s results remain largely intact even when a method designed for a small number of large clusters such as the placebo test is used. This also confirms \citeauthor{ibragimovmueller2016}'s (\citeyear{ibragimovmueller2016}) finding that most of \citeauthor{dalbofrechette2011}'s results continue to hold---albeit with larger $p$ values---if their method is applied. \sqed
\end{example}

In summary, the placebo test performs well even in fairly extreme (but empirically relevant) situations where the number of clusters is very small, the within-cluster correlation is high, and the clusters are very heterogeneous. These findings hold both for simulated and real data sets in a variety of estimation problems. 

\section{Conclusion}\label{s:conc}
I introduce a general, Fisher-style randomization testing framework to conduct nearly exact inference about the lack of effect of a binary treatment in the presence of very few, large clusters when the treatment effect is identified across clusters. The proposed randomization test formalizes and extends the intuitive notion of generating null distributions by assigning placebo treatments to untreated clusters. I show that under simple and easily verifiable conditions, the placebo test leads to asymptotically valid inference in a very large class of empirically relevant models. A simulation study and an empirical example are provided. The proposed inference procedure is easy to implement and performs well with as few as three treated and three untreated clusters.

\appendix

\section{Technical results}
This section states several auxiliary results that are used in the proofs in the next section. I start with an ``in probability'' version of a well-known result for deterministic sequences.

\begin{lemma}\label{l:pseq}
For each $\varepsilon > 0$, let $X_{n,\varepsilon}$ be random variable with $X_{n,\varepsilon}\pto 0$. Then there is a deterministic sequence $\varepsilon_n \to 0$ such that $X_{n,\varepsilon_n}\pto 0$. 
\end{lemma}

To formalize the concept of a permutation distribution, let $\pi$ be a random draw from a uniform distribution on $\Pi^*$ such that $\pi$ is independent of the cluster-level statistics $\hat{\theta}_n = (\hat{\theta}_{n,1}, \dots, \hat{\theta}_{n,q_n})$. Denote probability computed with respect to the randomness in $\pi$ by $\prob_\pi$. An object like $\prob_\pi(T_n(\pi\hat{\theta}_n)\leq t)$ is then the same as $\prob(T_n(\pi\hat{\theta}_n)\leq t \mid \hat{\theta}_n)$. The following result establishes that, after proper rescaling, the permutation distribution uniformly approximates the null distribution of $T_n$ in probability, conditional on the statistics $\hat{\theta}_n$. This is the standard measure of consistency for resampling distributions; see \citet[p.\ 329]{vandervaart1998}. Let \[ (t_1, \dots, t_{q_n})\mapsto S_n^2(t_1, \dots, t_{q_n}) = \frac{q_{1,n} q_{0,n}}{q_n} \hat{S}_n^2(t_1, \dots, t_{q_n}).\]
\begin{lemma}\label{l:weakconvergence}
If Assumption \ref{as:decom} holds with $\theta_1 = \theta_0$, then \[\sup_{t\in\mathbb{R}}\biggl|\prob\biggl(\sqrt{\frac{q_{1,n} q_{0,n}}{q_n}}\frac{T_n(\hat{\theta}_n)}{S_n(\hat{\theta}_n)}\leq t\biggr) - \prob_\pi\biggl(\sqrt{\frac{q_{1,n} q_{0,n}}{q_n}}\frac{T_n(\pi\hat{\theta}_n)}{S_n(\hat{\theta}_n)}\leq t\biggr)\biggr| \pto 0.\]
\end{lemma}

Let $g_n$ be nonrandom, measurable functions. A statement of the form ``$g_n(Y_n, X_n) \leadsto Z$ in probability, conditional on $X_n$'' means that the conditional distribution function of $g_n(Y_n, X_n)$ given $X_n$ converges in probability (pointwise at continuity points) to the distribution function of $Z$. I frequently use the following results in conjunction with the conditional Slutsky lemma \citep[Theorem 3.2]{xiongli2008} and with independent permutations $Y_n  = \pi$ so that $\prob_\pi(g_n(\pi, X_n) \leq z) = \prob(g_n(\pi, X_n) \leq z\mid X_n)$. These results are known in the literature but included here for easy reference.
\begin{lemma}\label{l:conditionalequiv}
Let $g_n$ be nonrandom, measurable functions. The following are equivalent:
\begin{enumerate}[\upshape (i)] 
\item $g_n(Y_n, X_n)\leadsto 0$ in probability, conditional on $X_n$,
\item $\prob(|g_n(Y_n, X_n)| > \varepsilon\mid X_n)\pto 0$ for all $\varepsilon > 0$, and
\item $g_n(Y_n, X_n)\pto 0$.
\end{enumerate}
\end{lemma}

The following result collects several useful consequences of Assumption \ref{as:decom} that will be used throughout the proofs in the next section. Let $1_{q_n}$ be a $q_n$-vector of ones and
\[\sigma_n^2 = \frac{q_{0,n}}{q_n q_{1,n}}\sum_{k=1}^{q_{1,n}}\ev Z_{n,k}^2 + \frac{q_{1,n}}{q_n q_{0,n}}\sum_{k=1+q_{1,n}}^{q_n} \ev Z_{n,k}.\]

\begin{lemma}\label{l:varapprox} Suppose Assumption \ref{as:decom} holds with $\theta_1 = \theta_0$. Then
	\begin{enumerate}[\upshape (i)]
	\item $S_n^2(\rn (\hat{\theta}_n-\theta_01_{q_n})) - \sigma_n^2 \pto 0$ and $S_n^2(\rn (\pi\hat{\theta}_n-\theta_01_{q_n})) - S_n^2(\pi Z_n) \pto 0$,
	\item there is a $\delta > 0$ such that $\prob(S_n(\pi Z_n) > \delta) \to 1$, and
	\item $S_n^2(\rn (\pi\hat{\theta}_n-\theta_01_{q_n}))/S_n^2(\pi Z_n)\pto 1$.
	\end{enumerate} 
\end{lemma}

The following two theorems are consequences of deep results on arrays with exchangeable weights due to \citet{janssen2005}. Bars denote averages over $q_n$ in the next two statements.
\begin{theorem}[{\citealt[Theorem~4.1]{janssen2005}}]\label{t:arrayclt} For arbitrary random arrays $X_{n,1},\dots, X_{n,q_n}$ with $\min\{q_{1,n},q_{0,n}\}\to \infty$ such that $0 < \liminf q_{1,n}/q_{0,n}\leq \limsup q_{1,n}/q_{0,n}< \infty$ and \[\frac{\max_{k\leq q_n}(X_{n,k} - \mean{X}_n)^2}{\sum_{k=1}^{q_n}(X_{n,k} - \mean{X}_n)^2}\pto 0 \]
we have, conditional on $X_{n,1},\dots X_{n,q_n}$,
\[\sqrt{\frac{q_{1,n} q_{0,n}}{q_n}} \frac{q_{1,n}^{-1}\sum_{k=1}^{q_{1,n}} X_{n,\pi(k)} - q_{0,n}^{-1}\sum_{k=1+q_{1,n}}^{q_n} X_{n,\pi(k)}}{S_n(X_{n,\pi(1)}, \dots, X_{n,\pi(q_n)})}\leadsto \mathrm{N}(0,1)\quad\text{in $\prob$-probability.} \] 
\end{theorem}

\begin{theorem}\label{t:arrayclt2}
Consider an exchangeable triangular array of real-valued weights $W_{n,1},\dots W_{n,q_n}$ such that {\upshape (i)} $\max_{k\leq q_n} |W_{n,k} - \mean{W}_n|\pto 0$, {\upshape (ii)} $\sum_{k=1}^{q_n} (W_{n,k} - \mean{W}_n)^2 \pto 1$, {\upshape (iii)} $\sqrt{q_n}(W_{n,1} - \mean{W}_n)\leadsto W$ with $\ev W = 0$ and $\var W = 1$. Also consider an arbitrary triangular array of real-valued random variables $X_{n,1},\dots X_{n,q_n}$ that is independent of the weights and satisfies {\upshape (iv)} $\max_{k\leq q_n} |X_{n,k}|\pto 0$, {\upshape (v)} $\sum_{k=1}^{q_n} (X_{n,k} - \mean{X}_n)^2 \pto \sigma^2$, and {\upshape (vi)} $\mean{X}_n\pto 0$. Then $\sqrt{q_n}\sum_{k=1}^{q_n} W_{n,k} (X_{n,k} - \mean{X}_n)$ converges in distribution to $\mathrm{N}(0,\sigma^2)$ in probability, conditional on $X_{n,1},\dots X_{n,q_n}$.	
\end{theorem}

Finally, the central limit theorems above do not require the $X_{n,k}$ to be independent. If they are, then the conditions of the theorem can be simplified using a result of A.~Y.~Khinchin. I include a proof for completeness.
\begin{lemma}[Khinchin]\label{l:arraymax}
A row-wise independent array $\{X_{n,k} : k\leq q_n, n\in\mathbb{N}\}$ satisfies $\max_{k\leq q_n} |X_{n,k}| \pto 0$ if and only if $\sum_{k=1}^{q_n}\prob(|X_{n,k}|> \varepsilon)\to 0$ for every $\varepsilon > 0$.
\end{lemma}

\section{Proofs}\label{s:proofs}

\begin{proof}[Proof of \eqref{eq:pval} and the equivalence of $T_n > c_{n,\alpha}$ and $p_n(\hat{\theta}_n) \leq \alpha$] To see the second equality in \eqref{eq:pval}, let $k_n$ be the index on the first order statistic of $\{T_n(\pi\hat{\theta}_n) : \pi\in\Pi\}$ strictly smaller than $T_n$. If no such order statistic exists, then $T_n = \min \{T_n(\pi \hat{\theta}_n) : \pi\in\Pi\}$ and the infimum in \eqref{eq:pval} must be $1$. Suppose therefore that $k_n$ exists and let $M = |\Pi|$. Then $p_n = p_n(\hat{\theta}_n)$ is the smallest $p$ such that $(k_n-1)/M < 1 - p \leq k_n/M$ \citep[see, e.g.,][p.\ 305]{vandervaart1998} or, equivalently, $p_n = 1-k_n/M$. Because $k_n = \sum_{\pi\in\Pi}1\{ T_{n}(\pi\hat{\theta}_n) < T_n \}$, the equality follows. 

To see the equivalence of $T_n > c_{n,\alpha}$ and $p_n \leq \alpha$, suppose first that $T_n > c_{n,\alpha}$. Then $\alpha \in \{ p : T_n > c_{n,p} \}$ and therefore $p_n =  \inf\{ p : T_n > c_{n,p} \}\leq \alpha$. Conversely, suppose $p_n \leq \alpha$. Note that $p \mapsto c_{n,p}$ is a right-continuous, decreasing step-function that moves through the order statistics of $\{T_n(\pi\hat{\theta}_n) : \pi\in\Pi\}$ in reverse order as $p$ increases. Hence, if $p_n\in[1/M,1)$, there is a $p$ such that $c_{n,p}$ is the first order statistic strictly smaller than $T_n$. The smallest such $p$ is $p_n$ and therefore $T_n > c_{n,p_n} \geq c_{n,\alpha}$. The extreme right-hand side of \eqref{eq:pval} contradicts $p_n < 1/M$ because $T_n \in \{T_n(\pi\hat{\theta}_n) : \pi\in\Pi\}$. Finally, $p_n = 1$ and $p_n \leq \alpha$ contradicts $\alpha \in (0,1)$.	\end{proof}

\begin{proof}[Proof of Lemma \ref{l:pseq}]
For each $m\in\mathbb{N}$, there is an $n_m$ such that $\prob(|X_{n,1/m}| > 1/m) < 1/m$ for all $n > n_m$. Without loss of generality, take $n_1<n_2<\cdots$. For each $n$, define implicitly $m_n$ as the $m$ that satisfies $n_m \leq n < n_{m+1}$ and let $m_n = 1$ for $n < n_1$. Then $\varepsilon_n = 1/m_n\to 0$ and, for any given $\delta >0$, we eventually have $\prob(|X_{n,\varepsilon_n}| > \delta) \leq \prob(|X_{n,\varepsilon_n}| > \varepsilon_n) < \varepsilon_n. $
\end{proof}

\begin{proof}[Proof of Proposition \ref{p:uniformremainder}]
Recall that $\lceil a \rceil$ is the smallest integer greater than $a$. Let $q_1(\varepsilon) = \lceil 1/\varepsilon\rceil$ and $q_0(\varepsilon) = \lceil q_1(\varepsilon)(1-\lambda)/\lambda \rceil$. By the continuous mapping theorem, \[X_{n,\varepsilon} := \frac{\sum_{k=1}^{q_1(\varepsilon)} |R_{n, k}| + \sum_{k=1 + q_1(\varepsilon)}^{q_1(\varepsilon) + q_0(\varepsilon)}|R_{n, k}|}{\sqrt{q_1(\varepsilon) + q_0(\varepsilon)}} \pto 0\] for every fixed $\varepsilon > 0$. By Lemma~\ref{l:pseq}, there is a sequence $\varepsilon_n \to 0$ such that $X_{n,\varepsilon_n}\pto 0$. Let $q_{1,n} = q_1(\varepsilon_n)$ and $q_{0,n} = q_0(\varepsilon_n)$. Note that $\min\{q_{1,n},q_{0,n}\}\to \infty$ and $q_{1,n}/(q_{1,n} + q_{0,n}) \to \lambda$, as desired.
\end{proof}

\begin{proof}[Proof of Lemma \ref{l:conditionalequiv}] (ii)~$\Rightarrow$~(iii): Apply the Lebesgue dominated convergence theorem. (ii)~$\Leftarrow$~(iii): The Markov inequality gives \[ \prob\bigl(\prob\bigl(|g_n(Y_n, X_n)| > \varepsilon\mid X_n\bigr) > \delta \bigr) \leq \delta^{-1} \prob\bigl(|g_n(Y_n, X_n)| > \varepsilon\bigr) \to 0. \] (i)~$\Rightarrow$~(ii): $\prob(|g_n(Y_n, X_n)| > \varepsilon\mid X_n)$ is at most
\[1- \prob\bigl(g_n(Y_n, X_n) \leq \varepsilon\mid X_n\bigr) + \prob\bigl(g_n(Y_n, X_n) \leq -\varepsilon\mid X_n\bigr).\] The definition of convergence in distribution in probability implies that the right-hand side converges to zero for every $\varepsilon > 0$. (i)~$\Leftarrow$~(iii): For every $z > 0$, $\prob(g_n(X_n, Y_n) \leq z\mid X_n) \pto 1$ and $\prob(g_n(X_n, Y_n) \leq -z\mid X_n) \pto 0$. Since $z$ is arbitrary, $z=0$ is a discontinuity point. The result follows.
\end{proof}

\begin{proof}[Proof of Lemma \ref{l:varapprox}]
Let $\mean{Z}_{1,n} = q_{1,n}^{-1}\sum_{k=1}^{q_{1,n}} Z_{n,k}$, $\mean{Z}_{0,n} = q_{0,n}^{-1}\sum_{k=1+q_{1,n}}^{q_n} Z_{n,k}$, $\mean{Z}_{1,n}(\pi) = q_{1,n}^{-1}\sum_{k=1}^{q_{1,n}} Z_{n,\pi(k)}$, and $\mean{Z}_{0,n}(\pi) = q_{0,n}^{-1}\sum_{k=1+q_{1,n}}^{q_n} Z_{n,\pi(k)}$. For (i), I start by approximating $S_n^2(\rn (\hat{\theta}_n-\theta_01_{q_n})) = \rn^2 S_n^2(\hat{\theta}_n)$. The absolute value of $\rn^2\sum_{k=1}^{q_{1,n}} (\hat{\theta}_{n,k} - \mean{\theta}_{1,n})^2 - \sum_{k=1}^{q_{1,n}} (Z_{n,k} - \mean{Z}_{1,n})^2$ is bounded above by a constant multiple of
\[ \biggl(\sum_{k=1}^{q_n} Z_{n,k}^2\biggr)^{1/2} \biggl(\sum_{k=1}^{q_n} R_{n,k}^2\biggr)^{1/2} + \sum_{k=1}^{q_n} R_{n,k}^2. \]
The same bound with $\sum_{k=1}^{q_{1,n}} (Z_{n,k} - \mean{Z}_{1,n})^2$ replaced by $\sum_{k=1+q_{1,n}}^{q_n} (Z_{n,k} - \mean{Z}_{0,n})^2$ applies to $\rn^2 \sum_{k=1+q_{1,n}}^{q_n} (\hat{\theta}_{n,k} - \mean{\theta}_{0,n})^2$; notice also that this bound is invariant to permutation. Use $\sum_{k=1}^{q_n} R_{n,k}^2/q_n\leq (\sum_{k=1}^{q_n} |R_{n,k}|)^2/q_n\pto 0$ to conclude that $S_n^2(\rn (\hat{\theta}_n-\theta_01_{q_n})) - S_n^2(Z_n)\pto 0$ as long as $q_n^{-1}\sum_{k=1}^{q_n} Z_{n,k}^2$ is bounded in probability. To this end, apply a weak law of large numbers for arrays \citep[Theorem 2.2.6, p.\ 52]{durrett2010} to see that $q_n^{-1}\sum_{k=1}^{q_n} (Z_{n,k}^2 - \ev Z_{n,k}^2)\pto 0$. The moment conditions are easily satisfied in view of $\smash{\sum_{k=1}^{q_n}}\ev |Z_{n,k}|^p = O(q_n)$ and
\begin{equation}\label{eq:zsqlln}
\biggl| \frac{1}{q_n}\sum_{k=1}^{q_n} \ev (Z_{n,k}^2 - \ev Z_{n,k}^2) 1\{Z_{n,k}^2 - \ev Z_{n,k}^2 > q_n\}\biggr| \leq \frac{2^p}{q_n^{p/2}}\sum_{k=1}^{q_n} \ev |Z_{n,k}|^p = O(q_n^{1-(p/2)}).	
\end{equation}
Chebyshev's inequality yields $\mean{Z}_{1,n} = q_{1,n}^{-1}\sum_{k=1}^{q_{1,n}} Z_{n,k}\pto 0$ and $\mean{Z}_{0,n} = q_{0,n}^{-1}\sum_{k=1+q_{1,n}}^{q_n} Z_{n,k}\pto 0$. To show the second part of (i), note that the bound for $\rn^2\sum_{k=1}^{q_{1,n}} (\hat{\theta}_{n,k} - \mean{\theta}_{1,n})^2$ also applies to $\rn^2\sum_{k=1}^{q_{1,n}} (\hat{\theta}_{n,\pi(k)} - \mean{\theta}_{1,n}(\pi))^2$ and $\rn^2\sum_{k=1+q_{1,n}}^{q_n} (\hat{\theta}_{n,\pi(k)} - \mean{\theta}_{0,n}(\pi))^2$. That bound converges to zero unconditionally in probability.

For (ii), we have $\ev (\mean{Z}_{1,n}(\pi)\mid \pi) = 0$ and, by the law of total variance, $\var(\mean{Z}_{1,n}(\pi))$ equals \[ q_{1,n}^{-2}\ev \sum_{k=1}^{q_{1,n}} \var(Z_{n,\pi(k)}\mid \pi)\leq q_{1,n}^{-2}\ev \sum_{k=1}^{q_n} \var(Z_{n,\pi(k)}\mid \pi) = q_{1,n}^{-2}\ev \var\biggl(\sum_{k=1}^{q_n} Z_{n,\pi(k)}\mid \pi\biggr).\] Because the sum on the far right is invariant to permutation, conclude that $\var(\mean{Z}_{1,n}(\pi)) \leq q_1^{-2} \sum_{k=1}^{q} \var Z_{n,k}\to 0$. It follows that $\mean{Z}_{1,n}(\pi)\pto 0$ and, by a similar argument, $\mean{Z}_{0,n}(\pi)\pto 0$ unconditionally. Use these two results, the law of large numbers for $q_n^{-1}\sum_{k=1}^{q_n} Z_{n,k}^2$ from above, and the fact that $S_n^2(\pi Z_n)$ is at least as large as \[ \min\biggl\{ \frac{q_{1,n}}{q_n(q_{0,n}-1)}, \frac{q_{0,n}}{q_n(q_{1,n}-1)} \biggr\}\biggl(\sum_{k=1}^{q_n} Z_{n,k}^2 - q_{1,n}\mean{Z}_{1,n}(\pi)^2 - q_{0,n}\mean{Z}_{0,n}(\pi)^2 \biggr) \] to see that the probability that this expression differs by more than $\varepsilon/2$ from $\min\{(1-\lambda)/\lambda,\lambda/(1-\lambda)\}\sum_{k=1}^{q_n} \var Z_{n,k}/q_n$ converges to zero unconditionally for every $\varepsilon > 0$. By assumption, the latter expression is eventually larger than a small enough $\varepsilon$. Conclude that $\prob(S^2_n(\pi Z_n) > \varepsilon/2)$ approaches one unconditionally.

Arguing along the same lines, we have $S^2_n(\pi Z_n)\leq \max\{q_{0,n}/(q_{1,n}-1), q_{1,n}/(q_{0,n}-1) \}\sum_{k=1}^{q_n} Z^2_{n,k}/q_n$ and therefore $S^2_n(\pi Z_n)$ is unconditionally bounded in probability. Conclude that $1-S_n^2(\rn (\hat{\theta}_n-\theta_01_{q_n}))/S_n^2(\pi Z_n)$ converges to zero conditionally in probability.
\end{proof}

\begin{proof}[Proof of Lemma \ref{l:arraymax}]
Apply the inequality $\prob(\max_{k\leq q_n} X_{n,k}^2 > \varepsilon ) \leq \sum_{k=1}^{q_n} \prob(X_{n,k}^2 > \varepsilon)$, then independence, and finally $\log x \leq x-1$ in
\[ 1 -\sum_{k=1}^{q_n} \prob(X_{n,k}^2 > \varepsilon) \leq \prob\Bigl(\max_{k\leq q_n} X_{n,k}^2 \leq \varepsilon \Bigr) = \prod_{k=1}^{q_n} \bigl(1 - \prob(X_{n,k}^2 > \varepsilon)\bigr) \leq e^{-\sum_{k=1}^{q_n} \prob(X_{n,k}^2 > \varepsilon)} \leq 1 \]
to obtain the desired result.
\end{proof}

\begin{proof}[Proof of Lemma \ref{l:weakconvergence}]
By Assumption \ref{as:decom}, $\sigma^2_n = O(1)$ and $\sigma^2_n \geq\min\{q_{0,n}/q_{1,n}, q_{1,n}/q_{0,n} \} \times \sum_{k=1}^{q_n} \var Z_{n,k}/q_n > 0$, which implies $1/\sigma_n = O(1)$. Let $w_{n,k} = \sqrt{q_{1,n}q_{0,n}}(q_{1,n}^{-1}1\{k \leq q_{1,n}\} - q_{0,n}^{-1}1\{k > q_{1,n}\})$. Use the fact that the limit superior of $\sup_{k\in\mathbb{N}} |w_{n,k}|$ is finite to deduce that
\begin{align}\label{eq:remainderbound}
\biggl|\frac{1}{\sqrt{q_n}}\sum_{k=1}^{q_n} w_{n,k} R_{n,k}\biggr| \leq \frac{\sup_{k\in\mathbb{N}}|w_{n,k}|}{\sqrt{q_n}}\sum_{k=1}^{q_n} |R_{n, k}|
\end{align}
converges to zero in probability and therefore
\begin{equation*}
\sqrt{\frac{q_{1,n} q_{0,n}}{q_n}}\frac{T_n(\rn (\pi\hat{\theta}_n-\theta_01_{q_n}))}{\sigma_n} = \frac{1}{\sqrt{q_n}}\sum_{k=1}^{q_n} \frac{w_{n,k} \rn(\hat{\theta}_{n,k} - \theta_0)}{\sigma_n} = \frac{1}{\sqrt{q_n}}\sum_{k=1}^{q_n} \frac{w_{n,k} Z_{n,k}}{\sigma_n} + \op (1).
\end{equation*}
Assumption \ref{as:decom}(iii) ensures that the first term on the right of satisfies the Lyapunov condition. The Lindeberg-Feller central limit theorem and the Slutsky lemma then yield weak convergence of the display to $\mathrm{N}(0, 1)$. By Lemma \ref{l:varapprox}, \[\sqrt{\frac{q_{1,n} q_{0,n}}{q_n}}\biggl(\frac{T_n(\rn (\pi\hat{\theta}_n-\theta_01_{q_n}))}{S_n(\rn (\pi\hat{\theta}_n-\theta_01_{q_n}))} - \frac{T_n(\rn (\pi\hat{\theta}_n-\theta_01_{q_n}))}{\sigma_n}\biggr)\pto 0.\] 

I now turn to the permutation statistics. Define  \[ \tilde{T}_n(\pi\hat{\theta}_n) = \frac{T_n(\pi\hat{\theta}_n)}{S_n(\hat{\theta}_n)} \frac{S_n(\pi \hat{\theta}_n)}{S_n(\pi Z_n)} = \sum_{k=1}^{q_n} \frac{w_{n,k} (\hat{\theta}_{n,\pi(k)} - \theta_0)}{S_n(\pi Z_n)}. \] By Lemma \ref{l:varapprox}, we can work on the set $\{ S_n(\pi Z_n) > \delta \}$. Just like the original statistic, $\sqrt{q_{1,n}q_{0,n}/q_n}\rn \tilde{T}_n(\pi\hat{\theta}_n)$ can be decomposed into
\begin{equation*}
\frac{1}{\sqrt{q_n}}\sum_{k=1}^{q_n} \frac{w_{n,k} \rn(\hat{\theta}_{n,\pi(k)} - \theta_0)}{S_n(\pi Z_n)} = \frac{1}{\sqrt{q_n}}\sum_{k=1}^{q_n} \frac{w_{n,k}(Z_{n,\pi(k)} + R_{n,\pi(k)})}{S_n(\pi Z_n)}.	
\end{equation*}
The absolute value of $\sum_{k=1}^{q_n} w_{n,k}R_{n,\pi(k)}/\sqrt{q_n}$ is bounded above the right-hand side of~\eqref{eq:remainderbound}, which is invariant to permutation. Conclude that $\sqrt{q_{1,n}q_{0,n}/q_n}\rn \tilde{T}_n(\pi\hat{\theta}_n)$ is within $\op(1)$ of \[\frac{1}{\sqrt{q_n}}\sum_{k=1}^{q_n} \frac{w_{n,k}Z_{n,\pi(k)}}{S_n(\pi Z_n)}.\] The conditional asymptotic distributions of the display and $\sqrt{q_{1,n}q_{0,n}/q_n}\tilde{T}_n(\pi\hat{\theta}_n)$ therefore coincide by the conditional Slutsky lemma. However, because permutation renders the $Z_{n,\pi(k)}$ dependent, the Lindeberg-Feller central limit theorem is no longer appropriate. I therefore apply multiplier central limit theory from Theorem \ref{t:arrayclt}. In view of \eqref{eq:zsqlln}, it is enough to show $\mean{Z}_n\pto 0$ and $\max_{k\leq q_n}Z_{n,k}^2/q_n\pto 0$. The first condition is an immediate consequence of $\var \mean{Z}_n = O(q_n^{-1})$. By Lemma \ref{l:arraymax}, the second condition holds because $ \sum_{k=1}^{q_n}\prob(Z_{n,k}^2 > q_n \varepsilon) \leq \sum_{k=1}^{q_n} \ev |Z_{n,k}|^p/(q_n\varepsilon)^{p/2}\to 0$. It follows from Lemma \ref{l:conditionalequiv} and the conditional Slutsky lemma that $\sqrt{q_{1,n}q_{0,n}/q_n}\rn\tilde{T}_n(\pi\hat{\theta}_n) \leadsto Z \sim \mathrm{N}(0,1)$ in probability. Apply Lemma \ref{l:varapprox} and the conditional Slutsky lemma again to see that \[ \biggl(1 - \frac{S_n(\pi Z_n)}{S_n(\rn (\pi\hat{\theta}_n-\theta_01_{q_n}))}\biggr) \sqrt{\frac{q_{1,n}q_{0,n}}{q_n}}\rn \tilde{T}_n(\pi\hat{\theta}_n) \leadsto 0\times Z = 0 \] in probability conditionally and therefore $\sqrt{q_{1,n}q_{0,n}/q_n}\rn T_n(\pi\hat{\theta}_n) \leadsto \mathrm{N}(0,1)$ in probability conditionally. 

Use P\'olya's theorem and the triangle inequality to see that the desired conclusion follows if $\sup_{t\in\mathbb{R}}|\prob_\pi(\sqrt{q_{1,n}q_{0,n}/q_n}\rn T_n(\pi\hat{\theta}_n) \leq t) - \Phi(t)|\pto 0$, where $\Phi$ is the standard normal distribution function. But this is already a consequence of the pointwise conditional convergence by a well known result \citep[see, e.g,][Problem 23.1, p.\ 339]{vandervaart1998}; it can shown by constructing diagonal subsequences that ensure almost sure convergence uniformly on a countable dense subset of the real line.
\end{proof}

\begin{proof}[Proof of Theorem \ref{t:testconsistency}] Under the null, $\sqrt{q_{1,n}q_{0,n}/q_n}c_{n,\alpha}/S_n(\hat{\theta}_n)$ converges in probability to $\Phi^{-1}(1-\alpha)$ by the properties of quantile functions \citep[Lemma 21.2, p.\ 305]{vandervaart1998} and Lemma \ref{l:weakconvergence}. Then \[\sqrt{\frac{q_{1,n}q_{0,n}}{q_n}}\frac{T_n - c_{n,\alpha}}{S_n(\hat{\theta}_n)} \leadsto \mathrm{N}(0,1) - \Phi^{-1}(1-\alpha)\] by the Slutsky lemma and therefore $\prob(T_n > c_{n,\alpha}) \to \alpha.$
	
Consider the alternative $\beta > 0$. Suppose $\min\{q_{1,n},q_{0,n}\}\to \infty.$ Let $a_{n,k} = q_{1,n}^{-1}1\{k \leq q_{1,n}\} - q_{0,n}^{-1}1\{k > q_{1,n}\}$ and write $T_n = \beta + \sum_{k=1}^{q_n} a_{n,k} 	(Z_{n,k} + R_{n,k}) = \beta + \op(1)$ by arguments as in the proof of Lemma \ref{l:varapprox}. Turning to the permutation statistics, notice that $q_{1,n}^{-1}\sum_{k=1}^{q_{1,n}} \hat{\theta}_{n,\pi(k)}-\theta_0 - \beta q_{1,n}/q_n = \beta q_{1,n}^{-1}\sum_{k=1}^{q_{1,n}} (1\{\pi(k)\leq q_{1,n}\} - q_{1,n}/q_n ) + \op(1)$ using arguments from before. The variance of the first term is bounded above by $\beta^2$ times
\[ \frac{1}{q_{1,n}}  + \biggl(\frac{q_{1,n}}{q_n}\frac{q_{1,n}-1}{q_n-1} - \Bigl(\frac{q_{1,n}}{q_n}\Bigr)^2\biggr) \frac{q_{1,n}-1}{q_{1,n}} \to 0. \]
The same argument applies to $q_{0,n}^{-1}\sum_{k=1+q_{1,n}}^{q_n} \hat{\theta}_{n,\pi(k)}$. Conclude that $S_n^2(\pi\hat{\theta}_n)$ is at least as large as \[ \min\biggl\{ \frac{q_{1,n}}{q_n(q_{0,n}-1)}, \frac{q_{0,n}}{q_n(q_{1,n}-1)} \biggr\}\biggl(\sum_{k=1}^{q_n} \Bigl(\hat{\theta}_{n,k}-\theta_0- \beta \frac{q_{1,n}}{q_n}\Bigr)^2 - \op(q_{1,n}) - \op(q_{0,n}) \biggr). \]
Calculations as in the proof of Lemma \ref{l:varapprox} show that this is within $\op(1)$ of 
\[ 
\min\biggl\{ \frac{q_{1,n}}{q_{0,n}-1}, \frac{q_{0,n}}{q_{1,n}-1} \biggr\}\biggl(\frac{1}{q_n}  \sum_{k=1}^{q_n} \ev Z_{n,k}^2 + \beta^2\frac{q_{1,n}q_{0,n}^2 + q_{1,n}^2q_{0,n}}{q_n^3} \biggr), \] 
which is bounded away form zero for $n$ large enough. The right-hand side in \[ \frac{S_n^2(\hat{\theta}_n)}{S_n^2(\pi\hat{\theta}_n)} \leq \frac{\max\{ q_{1,n}/(q_{0,n}-1), q_{0,n}/(q_{1,n}-1) \}}{\min\{ q_{1,n}/(q_{0,n}-1), q_{0,n}/(q_{1,n}-1) \}} \frac{ q_n^{-1}\sum_{k=1}^{q_n} (\hat{\theta}_{n,k}-\theta_0- \beta q_{1,n}/q_n)^2 }{ q_n^{-1}\sum_{k=1}^{q_n} (\hat{\theta}_{n,k}-\theta_0- \beta q_{1,n}/q_n)^2 + \op(1) } \] therefore converges in probability unconditionally either to $((1-\lambda)/\lambda)^2$ or its reciprocal depending on whether $\lambda \leq .5$ or not. Conclude that 
\begin{equation}\label{eq:alternativequant}
T_n(\pi\hat{\theta}_n) = \biggl( \frac{1}{q_{1,n}}\sum_{k=1}^{q_{1,n}} \hat{\theta}_{n,\pi(k)} - \frac{1}{q_{0,n}}\sum_{k=1+q_{1,n}}^{q_n} \hat{\theta}_{n,\pi(k)} \biggr) \frac{S_n(\hat{\theta}_n)}{S_n(\pi\hat{\theta}_n)} \pto 0	
\end{equation}
unconditionally. It follows that for $0 < \varepsilon < \beta$ we have that  the first term on the right of \[\prob(T_n \leq c_{n,\alpha}) \leq \prob(T_n \leq \beta-\varepsilon) + \prob(c_{n,\alpha} > \beta -\varepsilon )\] converges to zero. By the properties of quantile functions \citep[Lemma 21.1(i), p.\ 304]{vandervaart1998}, the second term on the right is equal to $\prob( \prob_\pi (T_n(\pi\hat{\theta}_n) > \beta-\varepsilon) > \alpha).$ Because this converges to 0 by \eqref{eq:alternativequant} and Lemma \ref{l:conditionalequiv}, conclude that $\prob(T_n > c_{n,\alpha}) \to 1$.

Now consider the alternative $\beta > 0$ with $q_1$ and $q_0$ fixed. Let $\pi_0 = (1,2,\dots, q)$. We have $\hat{\theta}_{n,k} \pto \theta_0 + \beta 1\{ k\leq q_1 \}$ and therefore $T_n\pto \beta$. For $0 < \varepsilon < \beta$, use this result in
$\prob( T_n \leq c_{n,\alpha}) \leq \prob(T_n \leq \beta -\varepsilon) + \prob(c_{n,\alpha} > \beta -\varepsilon)$
to see that the first term on the right converges to zero. By the properties of quantile functions, the second term can be written as
\[\prob\biggl( \sum_{\pi\in\Pi} 1\biggl\{ \frac{\hat{S}_n(\hat{\theta}_n)}{\hat{S}_n(\pi\hat{\theta}_n)}\sum_{k=1}^q a_{\pi(k)}\hat{\theta}_{n,k} > \beta -\varepsilon \biggr\} > \alpha |\Pi|\biggr).  \]
For each $\pi$, $\sum_{k=1}^q a_{\pi^{-1}(k)}\hat{\theta}_{n,k}$ is within $\op(1)$ of $\beta\sum_{k=1}^{q_1}a_{\pi^{-1}(k)}.$ Note that $\sum_{k=1}^{q_1}a_{k} = 1$ and $\beta\sum_{k=1}^{q_1}a_{\pi^{-1}(k)} < \beta$ for every $\pi\neq \pi_0$ because $\sum_{k=1}^{q_1}a_{\pi^{-1}(k)}$ can at most be $q_1^{-1}(q_1 - 1) - q_0^{-1}$. Let $p_1(\pi) = q_1^{-1}\sum_{k=1}^{q_1}1\{ \pi^{-1}(k) \leq q_1 \}$ and $p_0(\pi) = q_0^{-1}\sum_{k=1+q_1}^{q}1\{ \pi^{-1}(k) \leq q_1 \}$. Then \[\hat{S}_n^2(\pi\hat{\theta}_n)\pto \frac{p_1(\pi)\bigl(1-p_1(\pi)\bigr)}{q_1 - 1} + \frac{p_0(\pi)\bigl(1-p_0(\pi)\bigr)}{q_0 - 1}\] and, in particular, $\hat{S}_n(\hat{\theta}_n)\pto 0$. For $\pi \neq \pi_0$, the right-hand side of the display is strictly positive. Hence, $\prob( \sum_{k=1}^q a_{\pi^{-1}(k)}\hat{\theta}_{n,k}\hat{S}_n(\hat{\theta}_n)/\hat{S}_n(\pi \hat{\theta}_n) > \beta -\varepsilon)$ converges to $1\{ \pi = \pi_0 \}$ for every $\pi \neq \pi_0$ and the union bound gives $\prob( \cup_{\pi\neq \pi_0} \{\sum_{k=1}^q a_{\pi^{-1}(k)}\hat{\theta}_{n,k}\hat{S}_n(\hat{\theta}_n)/\hat{S}_n(\pi \hat{\theta}_n) > \beta -\varepsilon\} )\to 0$. Conclude that 
\[\prob\biggl( \sum_{\pi\in\Pi} 1\biggl\{ \frac{\hat{S}_n(\hat{\theta}_n)}{\hat{S}_n(\pi\hat{\theta}_n)}\sum_{k=1}^q a_{\pi^{-1}(k)}\hat{\theta}_{n,k} > \beta -\varepsilon \biggr\} = 1 \biggr) \to 1  \]
and therefore $\prob(c_{n,\alpha} > \beta -\varepsilon) \to 1\{ 1 > \alpha |\Pi| \}$, which proves the result for $|\Pi|^{-1} \leq \alpha$. If $|\Pi|^{-1} > \alpha$, then $1\{T_n > c_{n,\alpha}\} \equiv 0$ and therefore $\prob(T_n > c_{n,\alpha}) \to 0$ trivially. The result follows.
\end{proof}

\begin{proof}[Proof of Corollary \ref{c:testconsistencyeq}]
This can be shown following along the same lines as the proofs of Lemma \ref{l:weakconvergence} and Theorem \ref{t:testconsistency} with two modifications: both statistics can now be standardized by $\sigma_n$ and the appeal to Theorem \ref{t:arrayclt} has to be replaced by an appeal to Theorem \ref{t:arrayclt2}. Given the results in the proofs of Lemma \ref{l:weakconvergence} and Theorem \ref{t:testconsistency}, it suffies to verify conditions (i)-(iii) of Theorem \ref{t:arrayclt2} under the null. Recall $w_{n,k} = \sqrt{q_{1,n}q_{0,n}}(q_{1,n}^{-1}1\{k \leq q_{1,n}\} - q_{0,n}^{-1}1\{k > q_{1,n}\})$ and let $W_{n,k} = w_{n,\pi(k)}/\sqrt{q_n}$. Then $\mean{W}_n = 0$ and (i) follows immediately. Condition (ii) is satisfied because $W_{n,k}^2 = 1\{\pi(k) \leq q_{1,n}\}q_{0,n}/(q_{1,n}q_n) + 1\{\pi(k) > q_{1,n}\}q_{1,n}/(q_{0,n}q_n)$ and therefore the  $W_{n,k}^2$ sum to one. Because $w_{n,\pi(1)}$ converges weakly to a Rademacher variable, (iii) holds and the desired result follows. The result under the alternative is identical to the proof of the second part of Theorem \ref{t:testconsistencyfin} below.
	\end{proof}

\begin{proof}[Proof of Theorem \ref{t:testconsistencyfin}]
Let $Z = (Z_1,\dots,Z_{q})$ and note that the elements of $\Pi$ are now nonrandom. For $\pi\in\Pi$ and $x = (x_1,\dots, x_q)\in\mathbb{R}^{q}$, consider the map $x \mapsto T(x) = q_1^{-1}\sum_{k=1}^{q_1} x_k - q_0^{-1}\sum_{k=1+q_1}^{q} x_k$. Because $T(\pi \theta) = T(\pi\theta - \theta_0 1_q)$ and the test decision is invariant to multiplication by a positive scalar, I can work with $\rn(\hat{\theta}_n - \theta_0 1_q)$ in place of $\hat{\theta}_n$.

Suppose $\beta = 0$. Compared to the randomized test with a test function as in \eqref{eq:randtest} (replace $c_{n,\alpha}$ by $\mean{c}_{n,\alpha}$), we have $1\{ T_n > \mean{c}_{n,\alpha} \} \leq 1\{\varphi_{n,\alpha}(T_n)\geq U\}$ and therefore $\prob(T_n > \mean{c}_{n,\alpha}) \leq \prob\bigl(\varphi_{n,\alpha}(T_n)\geq U\bigr)$. The idea is now to apply Theorem 3.1 of \citet{canayetal2014} to $\prob\bigl(\varphi_{n,\alpha}(T_n)\geq U\bigr)$. Their Assumptions 3.1(i) and (ii) hold by assumption. Their Assumption 3.1(iii) holds if for any two $\pi, \pi'\in\Pi$ with $\pi\neq \pi'$ we either have $T(\pi x) = T(\pi' x)$ for all $x$ or $\prob(T(\pi Z) =  T(\pi' Z)) = 0$. Arguing similarly to the proof of \citet[Lemma S.5.1]{canayetal2014}, if that were not true, there would be $\pi, \pi'\in\Pi$ with $\pi\neq \pi'$ such that $T(\pi x) \neq T(\pi' x)$ for some $t$ and $\prob(T(\pi Z) =  T(\pi' Z)) > 0$. For this choice of $\pi \neq \pi'$, let $a_k = q_1^{-1}1\{ k \leq q_1 \} - q_0^{-1}1\{ k > q_1 \}$ and  $b_k = a_{\pi^{-1}(k)} - a_{\pi'^{-1}(k)}$. Because $T(\pi Z) = T(\pi' Z)$ if and only if $\sum_{k=1}^{q} b_k Z_k = 0$, conclude that $\prob(T(\pi Z) = T(\pi' Z)) > 0$ is equivalent to a discontinuity of the distribution of $\sum_{k=1}^{q} b_k Z_k$ at zero, \[ \prob \biggl( \sum_{k=1}^{q} b_k Z_k \leq 0 \biggr) > \prob \biggl( \sum_{k=1}^{q} b_k Z_k < 0 \biggr). \] A sum of independent variables is continuously distributed if at least one of the summands has that property. Because $\pi\neq \pi'$, we have $b_k\neq 0$ for at least one $k$. But every $Z_k$ is continuously distributed, which contradicts the preceding display. Conclude from continuity of $x\mapsto T(\pi x)$ for every given $\pi\in\Pi$ that \[ \limsup_{n\to\infty} \prob(T_n > \mean{c}_{n,\alpha}) \leq \lim_{n\to\infty} \prob\bigl(\varphi_{n,\alpha}(T_n)\geq U\bigr) = \alpha \] by Theorem 3.1 of \citet{canayetal2014}.

Now consider the alternative $\beta > 0$. As before, we have $\hat{\theta}_{n,k} \pto \theta_0 + \beta 1\{ k\leq q_1 \}$ and therefore $T_n\pto \beta$. For $0 < \varepsilon < \beta$, use this result in
$\prob( T_n \leq \mean{c}_{n,\alpha}) \leq \prob(T_n \leq \beta -\varepsilon) + \prob(\mean{c}_{n,\alpha} > \beta -\varepsilon)$
to see that the first term on the right converges to zero. By the properties of quantile functions, the second term can be written as
\[\prob\biggl( \sum_{\pi\in\Pi} 1\biggl\{ \sum_{k=1}^q a_{\pi^{-1}(k)}\hat{\theta}_{n,k} > \beta -\varepsilon \biggr\} > \alpha |\Pi|\biggr).  \]
For each $\pi$, $\sum_{k=1}^q a_{\pi^{-1}(k)}\hat{\theta}_{n,k}$ is within $\op(1)$ of $\beta\sum_{k=1}^{q_1}a_{\pi^{-1}(k)}.$ Note that $\sum_{k=1}^{q_1}a_{k} = 1$ and, in particular, $\beta\sum_{k=1}^{q_1}a_{\pi^{-1}(k)} < \beta$ for every $\pi\neq \pi_0$ because $\sum_{k=1}^{q_1}a_{\pi^{-1}(k)}$ can at most be $q_1^{-1}(q_1 - 1) - q_0^{-1}$. Hence, $\prob( 1\{\sum_{k=1}^q a_{\pi^{-1}(k)}\hat{\theta}_{n,k} > \beta -\varepsilon\} = 1 )$ converges to $1\{ \pi = \pi_0 \}$ for every $\pi \neq \pi_0$, possibly after decreasing $\varepsilon$, and the union bound gives $\prob( \cup_{\pi\neq \pi_0} \{\sum_{k=1}^q a_{\pi^{-1}(k)}\hat{\theta}_{n,k} > \beta -\varepsilon\} )\to 0$. Conclude that 
\[\prob\biggl( \sum_{\pi\in\Pi} 1\biggl\{ \sum_{k=1}^q a_{\pi^{-1}(k)}\hat{\theta}_{n,k} > \beta -\varepsilon \biggr\} = 1 \biggr) \to 1  \]
and therefore $\prob(\mean{c}_{n,\alpha} > \beta -\varepsilon) \to 1\{ 1 > \alpha |\Pi| \}$, which proves the result for $|\Pi|^{-1} \leq \alpha$. If $|\Pi|^{-1} > \alpha$, then $1\{T_n > \mean{c}_{n,\alpha}\} \equiv 0$ and therefore $\prob(T_n > \mean{c}_{n,\alpha}) \to 0$ trivially. The result follows.
\end{proof}

\bibliographystyle{chicago}
\bibliography{qspec.bib}

\end{document}